\documentclass[11pt]{article}
\usepackage{graphicx}
\usepackage[a4paper,margin=1in]{geometry}
\usepackage{amsmath,amssymb,amsthm,mathtools,bm,bbm}
\usepackage{algorithm}
\usepackage{algpseudocode}
\usepackage{enumitem}
\usepackage{tikz}
\usetikzlibrary{arrows.meta,positioning,calc,fit}
\usepackage{quantikz}
\usepackage{array}
\usepackage{hyperref}
\usepackage{float}
\usepackage{authblk} 
\usepackage{cite}

\hypersetup{colorlinks=true,linkcolor=blue,citecolor=blue,urlcolor=blue}

\newtheorem{definition}{Definition}[section]
\newtheorem{proposition}[definition]{Proposition}
\newtheorem{theorem}[definition]{Theorem}
\newtheorem{corollary}[definition]{Corollary}
\newtheorem{remark}[definition]{Remark}

\DeclareMathOperator{\Tr}{Tr}
\DeclareMathOperator{\rank}{rank}

\DeclareMathOperator{\Span}{span}

\DeclareMathOperator{\sech}{sech}

\newcommand{\R}{\mathbb{R}}
\newcommand{\I}{\mathbbm{1}}
\newcommand{\Herm}{\mathrm{Herm}}
\newcommand{\fd}[1]{S_{\mathrm{FD}}\!\left(#1\right)}
\newcommand{\ketbra}[1]{\lvert #1\rangle\!\langle #1\rvert}

\numberwithin{equation}{section}

\begin{document}
\title{\textbf{Quantum principal component analysis\\without eigenvector recovery}}
% \title{Quantum PCA without eigenvectors -- \\
% training-free recovery of principal scores without principal components}

\author[1]{Yewei Yuan} 
\author[2]{Michele Minervini}
\author[2]{Mark M. Wilde}
\author[1,3,4]{Nana Liu\thanks{nana.liu@quantumlah.org}} 

\affil[1]{Global College, Shanghai Jiao Tong University, Shanghai 200240, China}
\affil[2]{School of Electrical and Computer Engineering, \protect\\Cornell University, Ithaca, New York 14850, United States}

\affil[3]{Institute of Natural Sciences, Shanghai Jiao Tong University, Shanghai 200240, China}
\affil[4]{School of Mathematical Sciences, Shanghai Jiao Tong University, Shanghai 200240, China}

% ------------------
% \date{\small {\today}}
\date{\today}
\maketitle
\thispagestyle{empty} 
\pagestyle{plain}     
\begin{abstract}
Principal component analysis (PCA) is traditionally implemented through a covariance or kernel matrix, leading-eigenvector extraction, and hard rank-\(k\) projection. These steps can be computationally costly in high-dimensional and quantum-data settings, sensitive to small eigengaps, and unnecessary when downstream tasks only require principal-subspace scores. 
Such score-based objectives are important in applications such as anomaly detection, spectral-energy profiling, and other postselection tasks. To address these needs, we introduce a measurement-based soft PCA framework replacing the hard top-\(k\) projector with an entropy-regularized Fermi--Dirac filter. This filter is the unique optimizer of an entropy-regularized variational formulation of PCA and converges to the classical PCA projector in the zero-temperature limit. 

This filter has a direct interpretation as a quantum measurement, which naturally suggests a quantum approach. 
For centered covariance operators represented by quantum feature states, a single fixed circuit, together with threshold calibration, accesses all optimal filters for different rank budgets or retained-variance levels without rank-dependent circuit updates or eigenvector recovery. 
For new inputs, the same calibrated quantum circuit yields soft principal subspace scores, spectral energy profiles, and postselected filtered states. 
The required centering of both training and test data is performed coherently inside the quantum protocol, which is particularly important for quantum data where no classical feature vectors or centered Gram matrix are directly available.
By reframing PCA as a calibrated measurement task, this framework bypasses the need for iterative eigenvector extraction and achieves a \textit{dimension-independent} sample complexity $O(\eta^{-2})$ for normalized fractional-rank or retained variance scoring at additive accuracy $\eta$.
\end{abstract}

\newpage
\tableofcontents

\section{Introduction}

Principal component analysis (PCA) is one of the most standard ways to identify dominant structures and reduce the effective dimension in high-dimensional data~\cite{Jolliffe2002,JolliffeCadima2016}. Traditional PCA pipelines based on eigendecomposition can have several limitations: (i) keeping only the top $k$ eigenvectors imposes a hard spectral cut-off, so eigenmodes are either fully kept or completely discarded. Thus small perturbations in the data can change the selected subspace when the relevant eigengap is small; (ii) PCA requires the covariance matrix and kernel PCA requires the centered Gram matrix and its eigendecomposition, which can be costly to obtain in big data settings; (iii) even once the top eigenvectors are recovered, assigning a score to a new data point requires expensive projection onto the top eigenvectors. 

However, many practical applications care more about the final downstream task: how much of the new input data lies in the dominant covariance directions? This outputs a principal-subspace score. To answer this question directly, recovering dominant eigenvectors on their own can be overkill, and by avoiding this one can obviate some of the problems of traditional PCA above. One of the leading applications where only a final score is needed is anomaly detection\cite{jackson1979control,lakhina2004diagnosing,ding2013compressed}. Here if `normal' data live mostly in a leading principal subspace, then a new input is anomalous if it has low principal-subspace score. 
Similarly, PCA-subspace scores drive decision-making in classification tasks where individual subspaces are learned for each class, and test inputs are classified by comparing their projection scores or canonical angles to these subspaces~\cite{WatanabePakvasa1973,TurkPentland1991,Fukui2020SubspaceMethods}. In scientific data characterization and quality control, PCA scores are utilized to summarize dominant spectral or physical patterns, enabling the real-time analysis of complex processes like material synthesis~\cite{Kano2001,Marin2024RamanScoreMapping}. Furthermore, in model-order selection, the primary focus is evaluating how cumulative explained variance changes with the number of retained components. This effectively determines the minimal information required to represent complex systems without losing critical details~\cite{Minka2000,PerryPanigrahiBienWitten2025}.  
All of these tasks rely on extracting principal-subspace scores rather than solving explicit principal eigenvectors.

To directly solve this score-based PCA problem, a natural approach is to follow the philosophy of the variational characterization of leading eigenvalues~\cite{EckartYoung1936,Fan1951,Bhatia1997}, where the leading eigenspace can be rephrased as the solution of an optimization problem. For data with covariance matrix $C$, we can solve a trace-constrained measurement $M$ optimization problem (with solution $M^*$)
\begin{align*}
    \max_{M\in\mathcal M_d}\ \Tr(CM)
\qquad \text{subject to}\qquad
\Tr(M)=k, \qquad \mathcal M_d \coloneqq \{M\in \Herm(\mathbb C^d): 0\preceq M\preceq \I\}, 
\end{align*}
where $\text{Tr}(CM^*)$ is the retained covariance energy and $M^*$ can be shown to be the rank-\(k\) projector with a basic eigengap assumption, associated with traditional PCA that keeps only the top $k$ eigenvectors. The feasible set \(\mathcal M_d\cap\{M:\Tr(M)=k\}\) is the Fantope, the convex hull of rank-\(k\) projectors~\cite{Daspremont2007,Vu2013}. Although existing work has used this geometry mainly to impose sparsity or other structural constraints, here we use it to change the object returned by PCA itself. Once we get such measurement $M^*$, for downstream scoring of a normalized input \(\rho\), the corresponding principal-subspace score is \(\Tr(M\rho)\in[0,1]\). This allows us to directly output the score without needing to recover the eigenvectors.

This measurement perspective is particularly natural from the viewpoint of quantum computation, where $M$ being Hermitian and constrained to $0\preceq M\preceq \I$ can directly define a quantum measurement. Indeed, if the data is either natively quantum or can be embedded into quantum states, quantum algorithms like Quantum PCA~\cite{Lloyd2014} can be used to process the spectral information efficiently. In other settings, the covariance matrix operator may be directly available through state preparation or sampling, while writing down its full matrix or reconstructing its eigenvectors may be unnecessary~\cite{Gordon2022}. However, 
when the Fantope optimizer \(M^\star\) is the rank-\(k\) PCA projector, implementing it still amounts to resolving a sharp spectral cutoff. This typically requires eigenvector recovery by existing quantum approaches, which is overkill for score-based tasks. 
Furthermore, quantum implementations face an additional centering issue. For quantum-native data, or for nonlinear quantum feature maps, there may be no accessible classical feature vector to subtract and no explicit centered Gram matrix to build. Assuming classical centering for quantum data would require state tomography or explicit kernel reconstruction. 
Thus the question is now how to construct such a measurement so one can measure principal spectral content directly on new quantum feature states.   

We address this problem by softening the measurement optimization itself (adding a regularization), which is more readily solved and does not suffer from problems associated with a hard spectral cutoff. We develop a soft PCA measurement primitive based on Fermi--Dirac filtering~\cite{lindsey2023fast}. At the variational level, adding Fermi--Dirac entropy to the trace-constrained PCA objective replaces binary principal component occupations by smooth logistic occupations. The resulting operator is a temperature-controlled soft filter: finite temperature gives a stable spectral resolution, while the zero-temperature limit recovers hard PCA whenever the target eigengap is nonzero. This connects the Fantope view of PCA with the recent framework of quantum Fermi--Dirac thermal measurements~\cite{LiuWilde2026}.

We further show that this soft PCA filter can be calibrated without iterative quantum training.  A modified thermal measurement moves the scalar chemical potential (arising as a Lagrange multiplier to the constrained optimization problem) from the data-side Hamiltonian to a classical threshold on the control position. With a maximally mixed probe, the tail probability of the position outcome is exactly the normalized trace of the soft filter. A single fixed circuit therefore produces a position distribution whose quantiles identify the optimal thresholds for all target ranks. 
The same approach also applies to retained-variance tasks. 
Classically, these tasks are usually solved by first estimating the principal spectrum and then finding the smallest number of components that retains a prescribed fraction of covariance energy.
For this purpose, we introduce a covariance-probe version of the same fixed measurement. The full retained-variance profile can also be calibrated from one fixed position distribution, without reconstructing eigenvectors or eigenvalues.
After calibration, the same circuit can be applied to a centered test state to estimate normalized soft spectral-energy profiles or cumulative principal-subspace scores. 
The statistical cost is task dependent. Uniform all-rank calibration has a worst-case \(O(d^2/\eta^2)\) sample scaling for trace accuracy \(\eta\). A single fixed-rank target admits sharper binomial bounds, while normalized fractional-rank calibration and retained-variance calibration can achieve complexity \(O(\eta^{-2})\) independent of \(d\).

Our paper proceeds as follows. Sec.~\ref{sec:classical_pca} reformulates classical PCA as a trace-constrained measurement optimization over the Fantope.
Sec.~\ref{sec:2} softens this projector via Fermi–Dirac entropy regularization to derive a closed-form filter, and proves its convergence to hard PCA at low temperatures. 
Sec.~\ref{sec:3} presents the quantum construction: after representing centered feature covariances directly via quantum states, we deploy a modified thermal measurement. By moving the chemical potential to a classical position threshold, we achieve a training-free calibration for any target rank or retained variance, ready for inference on new data. Sec.~\ref{sec:comparison} positions the resulting protocol against existing classical and quantum PCA algorithms and identifies the downstream tasks where this new framework is favorable. Finally, Sec.~\ref{sec:4} discusses operational implications, statistical costs, limitations, and future directions.

\section{Classical PCA as Measurement Optimization}
\label{sec:classical_pca}
Let $\mathbf x_1,\dots,\mathbf x_N\in\R^d$ be classical data points, and let
\begin{equation}
\bar{\mathbf x} \coloneqq \frac{1}{N}\sum_{i=1}^N \mathbf x_i
\end{equation}
be their sample mean. Define the centered data
\begin{equation}
\mathbf y_i \coloneqq \mathbf x_i-\bar{\mathbf x},
\qquad i=1,\dots,N,
\end{equation}
and the sample covariance matrix
\begin{equation}
C \coloneqq \frac{1}{N}\sum_{i=1}^N \mathbf y_i\mathbf y_i^\top \in \R^{d\times d}.
\end{equation}
Suppose that $C$ admits the spectral decomposition
\begin{equation}
C=\sum_{j=1}^d \lambda_j \mathbf u_j\mathbf u_j^\top,
\qquad
\lambda_1\ge \lambda_2\ge \cdots \ge \lambda_d\ge 0,
\end{equation}
where $\mathbf u_1,\dots,\mathbf u_d$ form an orthonormal eigenbasis. When some eigenvalues are degenerate (e.g., $\lambda_k = \lambda_{k+1}$), the choice of the leading $k$ eigenvectors is not unique. Nevertheless, any valid selection spans an optimal $k$-dimensional principal subspace
\begin{equation}
\mathcal U_k \coloneqq \Span\{\mathbf u_1,\dots,\mathbf u_k\}.
\end{equation}
The orthogonal projector onto $\mathcal U_k$ is always a rank-$k$ projector, given by
\begin{equation}
P_k \coloneqq \sum_{j=1}^k \mathbf u_j\mathbf u_j^\top.
\end{equation}
For centered data, PCA acts by the filtering map $\mathbf y \longmapsto P_k \mathbf y$, while for the original data it gives the rank-$k$ reconstruction $\mathbf x \longmapsto \bar{\mathbf x}+P_k(\mathbf x-\bar{\mathbf x})$.

The same PCA projector can also be characterized variationally. We now extend this representation to complex spaces for quantum data, and then express it as an optimization form in terms of quantum measurement operators.

\begin{proposition}[Classical PCA as trace-constrained measurement optimization]
\label{prop:classical-pca-meas-opt}
Let
\begin{equation}
\mathcal M_d \coloneqq \{M\in \Herm(\mathbb C^d): 0\preceq M\preceq \I\}
\end{equation}
denote the set of measurement operators acting on $\mathbb C^d$. For generality, suppose that the sample covariance matrix $C\in\Herm(\mathbb C^d)$ has eigenvalues $\lambda_1\ge \cdots \ge \lambda_d$. Consider
\begin{equation}
\label{eq:pca-meas-opt}
\max_{M\in\mathcal M_d}\ \Tr(CM)
\qquad \text{subject to}\qquad
\Tr(M)=k,
\end{equation}
where $k\in\{1,\dots,d-1\}$. Then the set of optimal measurement operators forms a convex set, whose extreme points are exactly the rank-$k$ orthogonal projectors onto the optimal principal subspaces of $C$. Moreover, if the eigengap condition $\lambda_k>\lambda_{k+1}$ holds, the optimizer is unique and given by
\begin{equation}
M^\star=P_k=\sum_{j=1}^k \ketbra{u_j},
\end{equation}
where $u_1,\dots,u_d$ are the orthonormal eigenvectors of $C$ corresponding to the ordered eigenvalues.
\end{proposition}

\begin{proof}
Let
\begin{equation}
C=\sum_{j=1}^d \lambda_j\ketbra{u_j}
\end{equation}
be the eigendecomposition of $C$. For any feasible matrix $M \in \mathcal{M}_d$, we define $m_j \coloneqq \bra{u_j}M\ket{u_j}$ for each $j$. From the trace constraint and $0 \preceq M \preceq \I$, these diagonal elements must satisfy
\begin{equation}
0\le m_j\le 1,
\qquad
\sum_{j=1}^d m_j=k.
\end{equation}
We evaluate the objective function by expanding the trace in the eigenbasis of $C$:
\begin{equation}
\Tr(CM) = \sum_{j=1}^d \lambda_j \bra{u_j}M\ket{u_j}=\sum_{j=1}^d \lambda_j m_j.
\end{equation}
Since the eigenvalues $\lambda_j$ are sorted in non-increasing order, the maximum possible value of this sum is $\sum_{j=1}^k \lambda_j$. The optimal solutions for $m_j$ form a convex polytope. Its extreme points correspond to integer assignments $m_j \in \{0,1\}$ where exactly $k$ variables equal $1$ (matching the largest eigenvalues) and the rest are $0$.

For any such extreme point assignment, because $0 \preceq M \preceq I$, the condition $\bra{u_j}M\ket{u_j} = 1$ forces $\ket{u_j}$ to be an eigenvector of $M$ with eigenvalue $1$. Similarly, $\bra{u_j}M\ket{u_j} = 0$ implies $\ket{u_j}$ is an eigenvector with eigenvalue $0$. Thus, any extreme-point optimizer is rigorously determined as an orthogonal projector onto a $k$-dimensional principal subspace of $C$.

Finally, if the eigengap condition $\lambda_k>\lambda_{k+1}$ holds, the optimal assignment is strictly unique:
\begin{equation}
m_1=\cdots=m_k=1,
\qquad
m_{k+1}=\cdots=m_d=0.
\end{equation}
In this case, the convex set of optimizers collapses to a single point, making the optimal operator strictly unique and given by the projector
\begin{equation}
M^\star = \sum_{j=1}^k \ketbra{u_j} = P_k,
\end{equation}
thus concluding the proof.
\end{proof}

The feasible set $\mathcal{M}_d \cap \{M : \Tr(M)=k\}$ is known as the Fantope, the convex hull of rank-$k$ projectors \cite{Daspremont2007,Vu2013}. Thus, Proposition~\ref{prop:classical-pca-meas-opt} serves as a measurement-operator analogue of Ky Fan's maximum principle. 
In classical machine learning, this framework is primarily studied as a convex relaxation of PCA to bypass non-convex rank constraints (a perspective that notably reduces Euclidean reconstruction error, as detailed in Appendix~\ref{app:fantope-shrinkage}). For the quantum framework developed in this paper, this Fantope structure plays a far more fundamental role. Any operator $M \in \mathcal{M}_d$ intrinsically defines a valid quantum measurement effect. Consequently, for a normalized input quantum state $\rho$, the quantity $\Tr(M\rho)$ directly corresponds to a physically measurable Born probability in $[0,1]$ serving as the normalized principal-subspace score. 

\section{Soft PCA and Fermi--Dirac Filter}
\label{sec:2}
We now soften the trace-constrained PCA measurement problem in Eq.~\eqref{eq:pca-meas-opt} by adding Fermi--Dirac entropy regularization. 
The hard rank-$k$ projector of classical PCA will be replaced by a temperature-dependent soft filter whose eigenvalues are Fermi--Dirac occupation numbers.

\begin{definition}[Fermi--Dirac entropy]
\label{def:FD-entropy}
For \(x\in(0,1)\), define the scalar binary entropy
\begin{equation}
h(x)
\coloneqq
-x\log x-(1-x)\log(1-x).
\end{equation}
We extend \(h\) continuously to \([0,1]\) by setting \(h(0)=h(1)=0\). 
For a measurement effect \(M\in\mathcal M_d\), define its Fermi--Dirac entropy:
\begin{equation}
\label{eq:FD-entropy}
\fd{M}
\coloneqq
\Tr h(M)
=
-\Tr[M\log M]-\Tr[(\I-M)\log(\I-M)].
\end{equation}
Equivalently, if \(m_1,\ldots,m_d\in[0,1]\) are the eigenvalues of \(M\), then
\begin{equation}
\fd{M}=\sum_{j=1}^d h(m_j).
\end{equation}
\end{definition}

The entropy-regularized soft PCA problem at temperature $T>0$ is
\begin{equation}
\label{eq:soft-pca-primal}
\max_{M\in\mathcal M_d}\ \Tr(CM)+T\fd{M}
\qquad \text{subject to}\qquad
\Tr(M)=k.
\end{equation}
The entropy term smooths the sharp rank-$k$ occupation pattern of hard PCA while keeping the feasible set unchanged.
For the trace-constrained problem Eq.~\eqref{eq:soft-pca-primal}, the dual variable is a scalar chemical potential $\mu\in\R$. For later use, define the Hermitian operator
\begin{equation}
\label{eq:A-of-mu-abstract}
H(\mu)\coloneqq \mu\I-C.
\end{equation}
The corresponding Fermi--Dirac filter is
\begin{equation}
\label{eq:thermal-projector-def}
M_{T,\mu}
\coloneqq
M_T\bigl(H(\mu)\bigr)
\coloneqq
\left(\I+e^{H(\mu)/T}\right)^{-1}
=
\left(\I+e^{(\mu\I-C)/T}\right)^{-1}.
\end{equation}

\begin{proposition}[Soft PCA as a Fermi--Dirac filter]
\label{thm:soft-pca-scalar-specialization}
The regularized problem Eq.~\eqref{eq:soft-pca-primal} has a unique optimizer of the Fermi--Dirac filter form   
\begin{equation}
\label{eq:soft-projector-form}
M_{T,\mu^\star}
=
\left(\I+e^{(\mu^\star\I-C)/T}\right)^{-1},
\end{equation}
where $\mu^\star\in\R$ is the unique solution of the trace equation
\begin{equation}
\label{eq:trace-equation-main}
\Tr\!\left[\left(\I+e^{(\mu\I-C)/T}\right)^{-1}\right]=k.
\end{equation}
Equivalently, $\mu^\star$ is the unique minimizer of the scalar dual function
\begin{equation}
\label{eq:scalar-dual-objective}
\Phi_T(\mu)
\coloneqq
\mu k + T\Tr\log\!\left(\I+e^{(C-\mu\I)/T}\right),
\end{equation}
which is a strictly convex scalar function, and whose derivative and second derivative are
\begin{equation}
\label{eq:scalar-dual-derivative}
\Phi_T'(\mu)=k-\Tr(M_{T,\mu}),
\qquad
\Phi_T''(\mu)=\frac{1}{T}\Tr\bigl[M_{T,\mu}(\I-M_{T,\mu})\bigr]>0.
\end{equation}

\end{proposition}

\begin{proof}
For the Lagrange multiplier $\mu\in\R$, we have the Lagrangian
\begin{align}
\mathcal L_T(M,\mu)
&=
\Tr(CM)+T\fd{M}+\mu\bigl(k-\Tr(M)\bigr)\\
&=
\Tr\bigl((C-\mu\I)M\bigr)+T\fd{M}+\mu k
.\label{eq:lag2}
\end{align}
For fixed $\mu$, the Fr\'echet derivative with respect to $M$ over the open set $0\prec M\prec \I$ is
\begin{equation}
\nabla_M \mathcal L_T(M,\mu)
=
C-\mu\I + T\bigl(\log(\I-M)-\log M\bigr).
\end{equation}
Setting this derivative to zero yields
\begin{equation}
\log M - \log(\I-M)=\frac{C-\mu\I}{T},
\end{equation}
and therefore
\begin{equation}
M(\I-M)^{-1}=e^{(C-\mu\I)/T}.
\end{equation}
Solving for $M$ gives
\begin{equation}
M=M_{T,\mu}=\left(\I+e^{(\mu\I-C)/T}\right)^{-1}.
\end{equation}

To obtain the scalar dual objective, substitute $M_{T,\mu}$ into the Lagrangian. Equivalently, diagonalize $C=\sum_{j=1}^d \lambda_j\ketbra{u_j}$ and optimize mode by mode, the first two terms of $\mathcal L_T(M,\mu)$ in Eq.~\eqref{eq:lag2} becomes:

\begin{equation}
\sup_{0\le m\le 1}\ \left[(\lambda_j-\mu)m + T\bigl(-m\log m -(1-m)\log(1-m)\bigr)\right]
=
T\log\!\left(1+e^{(\lambda_j-\mu)/T}\right).
\end{equation}
Summing over $j$ and restoring the $\mu k$ term gives Eq.~\eqref{eq:scalar-dual-objective}. Differentiating yields
\begin{equation}
\Phi_T'(\mu)
=
 k - \sum_{j=1}^d \frac{1}{1+e^{(\mu-\lambda_j)/T}}
=
 k - \Tr(M_{T,\mu}),
\end{equation}
and
\begin{equation}
\Phi_T''(\mu)
=
\frac{1}{T}\sum_{j=1}^d
\frac{e^{(\mu-\lambda_j)/T}}{\bigl(1+e^{(\mu-\lambda_j)/T}\bigr)^2}
=
\frac{1}{T}\Tr\bigl[M_{T,\mu}(\I-M_{T,\mu})\bigr]>0,
\end{equation}
which is Eq.~\eqref{eq:scalar-dual-derivative}.
So $\Phi_T$ is strictly convex and admits a unique minimizer $\mu^\star$. The first order condition $\Phi_T'(\mu^\star)=0$ is exactly the trace equation Eq.~\eqref{eq:trace-equation-main}.
\end{proof}

From Proposition~\ref{thm:soft-pca-scalar-specialization}, if
\begin{equation}
C=\sum_{j=1}^d \lambda_j\ketbra{u_j},
\end{equation}
then
\begin{equation}
\label{eq:spectral-soft-projector}
M_{T,\mu}
=
\sum_{j=1}^d \frac{1}{1+e^{(\mu-\lambda_j)/T}}\,\ketbra{u_j}.
\end{equation}
Thus the hard occupation numbers $0$ and $1$ of classical PCA are replaced by Fermi--Dirac occupation weights in the interval $(0,1)$. In particular, for $T>0$ and $1\le k\le d-1$, the optimizer $M_{T,\mu^\star}$ is generally not idempotent, and therefore is a \emph{soft filter} rather than a hard rank-$k$ projector. We now show that Soft PCA recovers the standard hard projector as the parameter $T$ vanishes.

\begin{corollary}[Low-temperature limit]
\label{cor:temperature-choice}
For each $T>0$, let
\begin{equation}
M_T^\star \coloneqq M_{T,\mu^\star}
\end{equation}
denote the unique optimizer of the entropy-regularized soft PCA problem Eq.~\eqref{eq:soft-pca-primal}, and let
$v_{\mathrm{hard}}$
be the optimal value of the hard measurement optimization problem Eq.~\eqref{eq:pca-meas-opt}. Let \(h\) be the scalar binary entropy defined in Definition~\ref{def:FD-entropy}.
For every $\epsilon>0$, if
\begin{equation}
\label{eq:Topt}
T\le \frac{\epsilon}{d\,h(k/d)},
\end{equation}
then
\begin{equation}
\label{eq:uniform-soft-hard-gap}
\Tr(CM_T^\star)\ge v_{\mathrm{hard}}-\epsilon.
\end{equation}
Note that since
\begin{equation}
\label{eq:uneqt}
k\log(d/k)\le d\,h(k/d)\le k\log(d/k)+k,
\end{equation}
it follows that for practical scenarios where $k = O(1)$, the sufficient temperature scale in Eq.~\eqref{eq:Topt} is $O(1/\log d)$.
In addition, if the eigengap condition
$\lambda_k(C)>\lambda_{k+1}(C)$
holds, then
\begin{equation}
\lim_{T\to 0} M_T^\star = P_k.
\end{equation}

\end{corollary}

\begin{proof}
Since $M_T^\star$ is feasible for Eq.~\eqref{eq:pca-meas-opt}, one has
\begin{equation}
\label{eq:v_lower_bound}
\Tr(CM_T^\star)\le v_{\mathrm{hard}}.
\end{equation}
On the other hand, because $M_T^\star$ is an optimizer of Eq.~\eqref{eq:soft-pca-primal}, we have
\begin{equation}
\label{eq:v_uper_bound}
v_{\mathrm{hard}}
\le
\Tr(CM_T^\star)+T\fd {M_T^\star}
\end{equation}
Now let $m_1,\dots,m_d\in[0,1]$ be the eigenvalues of $M_T^\star$. Then
\begin{equation}
\sum_{j=1}^d m_j=\Tr(M_T^\star)=k,
\qquad
\fd {M_T^\star}=\sum_{j=1}^d h(m_j).
\end{equation}

Because $h$ is concave, Jensen's inequality gives
\begin{equation}
\fd {M_T^\star}
\le
d\,h\!\left(\frac{1}{d}\sum_{j=1}^d m_j\right)
=
d\,h(k/d).
\end{equation}
Therefore
\begin{equation}
v_{\mathrm{hard}}
\le
\Tr(CM_T^\star)+T\,d\,h(k/d),
\end{equation}
which proves Eq.~\eqref{eq:uniform-soft-hard-gap}.  The bound
\begin{equation}
k\log(d/k)\le d\,h(k/d)\le k\log(d/k)+k
\end{equation}
follows from the elementary inequality
\begin{equation}
x\log(1/x) \le h(x)\le x\log(1/x)+x,
\qquad x\in[0,1].
\end{equation}

Now let $T\to 0$, since the feasible set
\begin{equation}
\{M\in \mathcal M_d: \Tr(M)=k\}
\end{equation}
is compact and closed, by Eq.~\eqref{eq:v_lower_bound}, Eq.~\eqref{eq:v_uper_bound} and continuity of the trace,
\begin{equation}
\lim_{T\to 0}\Tr(CM_{T}^\star)= v_{\mathrm{hard}},
\end{equation}
which is the optimal value of Eq.~\eqref{eq:pca-meas-opt}.

If the eigengap condition $\lambda_k(C)>\lambda_{k+1}(C)$ holds, then Proposition~\ref{prop:classical-pca-meas-opt} implies that Eq.~\eqref{eq:pca-meas-opt} has the unique optimizer $P_k$, and
\begin{equation}
\lim_{T\to 0} M_T^\star = P_k,
\end{equation}
thus concluding the proof.
\end{proof}

\section{Quantum Fermi--Dirac Filter for Soft PCA}
\label{sec:3}
Section~\ref{sec:2} reformulated soft PCA as a trace-constrained Fermi--Dirac optimization, and Appendix~\ref{app:classicalopt} showed that its classical implementation still requires repeated evaluation of matrix-valued quantities such as
\begin{equation}
\Tr(M_{T,\mu}),
\qquad
M_{T,\mu}=\left(\I+e^{(\mu\I-C)/T}\right)^{-1},
\end{equation}
which becomes costly in high dimensions. 
For dense unstructured covariance operators, direct methods require \(O(d^3)\) spectral or matrix-function computations, motivating a quantum implementation in which \(C\) is accessed directly as an operator.

We first specify a quantum representation of the covariance operator. We then introduce a modified thermal measurement that transfers the scalar shift \(\mu\I\) from the data-side Hamiltonian to the control side. This leads to a training-free quantile calibration of the optimal soft PCA filter. Once the optimal soft PCA filter is obtained, the computation of the final score on new input states can be performed directly. 

\subsection{Quantum covariance representation}
\label{sec:qcovariance}
There are two natural settings in which the present formulation is especially well matched to quantum computation. First, when the data are quantum native, the uncentered second-moment is already a density operator, meaning the covariance object can be built directly from quantum states. Second, when classical data are encoded through a quantum feature map
\begin{equation}
\mathbf{x}\longmapsto \ket{\phi(\mathbf{x})}=U_\phi(\mathbf{x})\ket{0},
\end{equation}
the corresponding feature-space covariance can be formed directly at the level of quantum states, without first explicitly constructing a full Gram matrix. 
In both cases, the protocol operates on quantum samples and produces calibrated measurement effects that can be applied directly to new quantum input states.

\begin{proposition}[Feature-state covariance representation]
\label{prop:feature-state-covariance}
Let
\begin{equation}
\ket{\phi_i}\coloneqq \ket{\phi(\mathbf{x}_i)}\in\mathcal H,
\qquad i=1,\dots,N,
\end{equation}
be quantum feature states associated with training samples $\mathbf{x}_1,\dots,\mathbf{x}_N$, and define the empirical feature-state average
\begin{equation}
\bar\rho_\phi \coloneqq \frac{1}{N}\sum_{i=1}^N \ketbra{\phi_i}.
\end{equation}
Let
\begin{equation}
\ket{m}\coloneqq \frac{1}{N}\sum_{i=1}^N \ket{\phi_i}
\end{equation}
be the empirical mean feature vector. Then the centered feature-space covariance operator is
\begin{equation}
\label{eq:centered-feature-covariance}
C_\phi
\coloneqq
\frac{1}{N}\sum_{i=1}^N
\bigl(\ket{\phi_i}-\ket{m}\bigr)\bigl(\bra{\phi_i}-\bra{m}\bigr)
=
\bar\rho_\phi-\ketbra{m}.
\end{equation}
In particular, if the feature data are centered in Hilbert space, i.e. $\ket{m}=0$, then $C_\phi=\bar\rho_\phi$.

\end{proposition}

\begin{proof}
Expand the definition:
\begin{align}
C_\phi
&=
\frac{1}{N}\sum_{i=1}^N
\left(
\ketbra{\phi_i}
-\ket{\phi_i}\bra{m}
-\ket{m}\bra{\phi_i}
+\ketbra{m}
\right) \\
&=
\frac{1}{N}\sum_{i=1}^N \ketbra{\phi_i}
-\left(\frac{1}{N}\sum_{i=1}^N \ket{\phi_i}\right)\bra{m}
-\ket{m}\left(\frac{1}{N}\sum_{i=1}^N \bra{\phi_i}\right)
+\frac{1}{N}\sum_{i=1}^N \ketbra{m}.
\end{align}
Using
\begin{equation}
\frac{1}{N}\sum_{i=1}^N \ket{\phi_i}=\ket{m},
\qquad
\frac{1}{N}\sum_{i=1}^N \bra{\phi_i}=\bra{m},
\qquad
\frac{1}{N}\sum_{i=1}^N \ketbra{m}=\ketbra{m},
\end{equation}
we obtain
\begin{equation}
C_\phi=\bar\rho_\phi-\ketbra{m}.
\end{equation}
The centered case is immediate.
\end{proof}

It is important to distinguish the roles of the two terms in
\begin{equation}
\label{eq:centered-covariance-main}
C_\phi=\bar\rho_\phi-\ketbra{m}.
\end{equation}
The operator
\begin{equation}
\bar\rho_\phi=\frac{1}{N}\sum_{i=1}^N \ketbra{\phi_i}
\end{equation}
is a density operator on $\mathcal H$, since it is positive semidefinite and has unit trace. In general, it is a \emph{mixed state}, rather than a pure state.

By contrast,
\begin{equation}
\ketbra{m},
\qquad
\ket{m}=\frac{1}{N}\sum_{i=1}^N \ket{\phi_i},
\end{equation}
is a positive semidefinite with rank at most one, but it is not generally a density operator, because $\ket{m}$ need not be normalized. Indeed,
\begin{equation}
\Tr(\ketbra{m})=\braket{m}{m}=\alpha,
\end{equation}
and typically $\alpha\neq 1$.

This form for $C_\phi$ in Eq.~\eqref{eq:centered-covariance-main} is especially useful later, because it expresses $C_\phi$ as a linear combination of quantum states. In the special case $\alpha=0$, one simply has $\ket{m}=0$ and hence
\begin{equation}
C_\phi=\bar\rho_\phi.
\end{equation}

\begin{remark}
\label{rem:kernel}
To connect Proposition~\ref{prop:feature-state-covariance} with kernel PCA, we now consider the Gram representation. In classical kernel PCA, one diagonalizes the centered kernel matrix rather than the feature-space covariance operator. Let
$\Phi\coloneqq [\,\ket{\phi_1}\ \cdots\ \ket{\phi_N}\,]$
denote the feature matrix, and let
$H\coloneqq I_N-\frac{1}{N}\mathbf 1\mathbf 1^\top$
be the standard centering matrix. Using Proposition~\ref{prop:feature-state-covariance}, the centered covariance operator in quantum feature space can be written compactly as
\begin{equation}
C_\phi = \bar\rho_\phi-\ketbra{m} = \frac{1}{N}\Phi H\Phi^\dagger.
\end{equation}
Meanwhile, the empirical kernel matrix is
\begin{equation}
K = \Phi^\dagger\Phi,
\end{equation}
and its centered version is
\begin{equation}
K_c = H K H = H\Phi^\dagger\Phi H.
\end{equation}
Since the operators $\Phi H\Phi^\dagger$ and $H\Phi^\dagger\Phi H$ have the same nonzero spectrum, it follows that the nonzero eigenvalues of $C_\phi$ coincide with those of $\frac{1}{N}K_c$.
Therefore, centered quantum feature states induce exactly the same nonzero principal spectrum as classical kernel PCA, while representing this spectrum through the feature-space covariance operator rather than an explicitly formed Gram matrix. In this sense, \(C_\phi\) is the operator-level counterpart of the centered kernel-PCA matrix~\cite{Scholkopf1998}.
\end{remark}

We now describe the state preparation procedures for the two terms in Eq.~\eqref{eq:centered-covariance-main} and for arbitrary centered test input.

\paragraph{Preparation of $\bar\rho_\phi$.}
Assume access to a data loading unitary
\begin{equation}
\label{eq:data-loading-unitary}
U_{\mathrm{data}}\ket{i}\ket{0}=\ket{i}\ket{\phi_i},
\qquad i\in\{1,\dots,N\}.
\end{equation}
Let
\begin{equation}
\ket{u_N}\coloneqq \frac{1}{\sqrt N}\sum_{i=1}^N \ket{i}
\end{equation}
be the uniform superposition over training indices and define
\begin{equation}
\label{eq:Omega-state}
\ket{\Omega}
\coloneqq
U_{\mathrm{data}}\bigl(\ket{u_N}\ket{0}\bigr)
=
\frac{1}{\sqrt N}\sum_{i=1}^N \ket{i}\ket{\phi_i}.
\end{equation}
Tracing out the index register yields
\begin{equation}
\Tr_{\mathrm{idx}}\!\bigl[\ketbra{\Omega}\bigr]
=
\frac{1}{N}\sum_{i=1}^N \ketbra{\phi_i}
=
\bar\rho_\phi.
\end{equation}
Hence $\bar\rho_\phi$ is obtained by discarding the index register of the purification Eq.~\eqref{eq:Omega-state}.

\paragraph{Preparation of $\ket{m}$.}
Let $U_N$ be any unitary such that $U_N\ket{0}=\ket{u_N}$, such as a tensor product of Hadamard gates. Applying $U_N^\dagger$ to the index register of Eq.~\eqref{eq:Omega-state} gives
\begin{equation}
\label{eq:mean-state-postselection}
\ket{\chi_m}=(U_N^\dagger\otimes \I)\ket{\Omega}
=
\ket{0}\ket{m} +\ket{\perp},
\end{equation}
where the index component of $\ket{\perp}$ is orthogonal to $\ket{0}$, and
\begin{equation}
\alpha
=
\left\|(\bra{0}\otimes \I)(U_N^\dagger\otimes \I)\ket{\Omega}\right\|^2=
\langle m \vert m \rangle
\end{equation}
is exactly the probability of obtaining outcome $\ket{0}$ on the index register.
In the algorithms in this paper, however, we do not need to estimate \(\alpha\) in
advance. Rather than normalizing the mean branch and then multiplying it by
\(\alpha\), we use the subnormalized branch \(\ketbra m\) directly. The
index \(\ket0\) acts as an
internal selector for the mean term, so the factor \(\alpha\) enters
automatically through the branch amplitude.

In practice, when feature-space centering is available in advance, one may work directly with centered feature states and omit the mean state. If one instead chooses an uncentered PCA variant, the \(\ketbra m\) term can also be dropped. 
However, for the centered protocol considered here, classical zero-mean inputs do not generally imply zero mean in a nonlinear quantum feature space. Moreover, for quantum-native data, there may be no accessible classical feature vector whose mean can be subtracted in advance. We therefore use the same subnormalized mean branch to center new test inputs coherently.

\paragraph{Preparation of a centered input.}
For a new input \(\mathbf x\), assume access to a test state preparation unitary
\begin{equation}
U_{\mathbf x}\ket{0}
=
\ket{\phi(\mathbf x)}.
\end{equation}
The centered feature vector associated with \(\mathbf x\) is
\begin{equation}
\ket{z_{\mathbf x}}
\coloneqq
\ket{\phi(\mathbf x)}-\ket m,
\qquad
\nu_{\mathbf x}
\coloneqq
\braket{z_{\mathbf x}}{z_{\mathbf x}}.
\end{equation}
Using the same subnormalized mean branch
\(\ket{\chi_m}=\ket0\ket m+\ket{\perp}\) from
Eq.~\eqref{eq:mean-state-postselection}, one can prepare the centered input by
a simple LCU postselection \cite{ChildsWiebe2012,Berry2015}. Prepare
\begin{equation}
\ket{\Xi_{\mathbf x}}
=
\frac{1}{\sqrt2}
\left(
\ket0_a\ket0_{\rm idx}\ket{\phi(\mathbf x)}
+
\ket1_a\ket{\chi_m}
\right),
\end{equation}
where \(a\) is an LCU ancilla. Postselect the ancilla onto
\begin{equation}
\ket{-}_a
=
\frac{\ket0_a-\ket1_a}{\sqrt2}
\end{equation}
and the index register onto \(\ket0_{\rm idx}\). The resulting unnormalized input state is
\begin{equation}
\frac12
\left(
\ket{\phi(\mathbf x)}-\ket m
\right)
=
\frac12\ket{z_{\mathbf x}}.
\end{equation}
Therefore the success probability of this centered state preparation is
\begin{equation}
P_{\rm succ}(\mathbf x)
=
\frac14\nu_{\mathbf x},
\end{equation}
and, conditioned on success, the prepared state is
\begin{equation}
\ket{\widetilde z_{\mathbf x}}
=
\frac{\ket{z_{\mathbf x}}}{\sqrt{\nu_{\mathbf x}}},
\qquad
\nu_{\mathbf x}>0.
\end{equation}
Again, no prior estimate of \(\alpha=\braket m m\) is required, since the
subnormalized mean vector \(\ket m\) is produced directly as a branch amplitude.
If \(\nu_{\mathbf x}=0\), then the test input coincides with the feature-space
mean and its unnormalized centered spectral energy is zero.

\subsection{Modified quantum thermal measurement}
We now describe how the Fermi--Dirac soft filter
\begin{equation}
M_{T,\mu}=\left(\I+e^{(\mu\I-C)/T}\right)^{-1}
\end{equation}
is realized as a quantum measurement.
A direct thermal-measurement implementation~\cite{LiuWilde2026} would use a Hamiltonian evolution generated by \(\hat p\otimes H(\mu)\), where \((\hat q,\hat p)\) denotes position and momentum quadratures with $[\hat q,\hat p]=i$ and
\begin{equation}
H(\mu)=\mu\I-C.
\end{equation}
In the quantum setting, we take \(C=C_\phi\). Using Eq.~\eqref{eq:centered-covariance-main} we obtain
\begin{equation}
\label{eq:A-decomposition-two-step}
H(\mu)=\mu\I-\bar\rho_\phi+\ketbra{m}.
\end{equation}
Now introduce the maximally mixed state
\begin{equation}
\tau_{\mathrm{mm}}\coloneqq \frac{\I}{d}.
\end{equation}
Then
\begin{equation}
\label{eq:A-decomposition-final}
H(\mu)=d\mu\,\tau_{\mathrm{mm}}-\bar\rho_\phi+\ketbra{m}.
\end{equation}

If the first identity term in Eq.~\eqref{eq:A-decomposition-final} was represented within the same state-based Hamiltonian simulation, it would inflate the coefficient budget without adding covariance information. Its coefficient \(d\mu\) depends on the chemical potential variable $\mu$ and scales as \(O(d)\), while the informative covariance terms have coefficients of order one.

Since this term commutes with the covariance terms and acts only as a scalar shift, we move it to the control side, equivalently to the final position threshold. The data-side evolution then only needs to implement the covariance operator \(C_\phi\).

\begin{proposition}[Modified thermal measurement]
\label{prop:modified-thermal measurement}
Fix a temperature decomposition
\begin{equation}
T=T_1T_2,
\qquad
T_1>0,\quad T_2>0.
\end{equation}
Let
\begin{equation}
g_{T_1}(q)
\coloneqq
\frac{e^{-q/T_1}}{T_1\bigl(1+e^{-q/T_1}\bigr)^2},
\qquad q\in\mathbb R,
\end{equation}
and define the displaced control qumode state
\begin{equation}
\ket{\psi_{T_1,\delta}}
\coloneqq
\int_{\mathbb R}\sqrt{g_{T_1}(q-\delta)}\,\ket{q}\,dq,
\qquad \delta\in\mathbb R.
\end{equation}
Consider the following procedure:
\begin{enumerate}[label=\arabic*.]
\item prepare the control qumode in the state \(\ket{\psi_{T_1,\delta}}\);
\item prepare the data register in a probe state \(\rho_{\mathrm{in}}\);
\item apply the joint evolution
\begin{equation}
\label{eq:modified-Hevolution}
U_C\coloneqq e^{-i\hat p\otimes C/T_2};
\end{equation}
\item measure the control register in the position basis, obtaining \(q\in\mathbb R\);
\item output \(1\) if \(q>\beta\) and \(0\) otherwise, where \(\beta\in\mathbb R\) is a fixed threshold.
\end{enumerate}
Then the resulting binary POVM is
\begin{equation}
\bigl(\widetilde M_{T,\delta,\beta},\,\I-\widetilde M_{T,\delta,\beta}\bigr),
\end{equation}
with
\begin{equation}
\label{eq:modified-effect}
\widetilde M_{T,\delta,\beta}
=
\left(\I+e^{(((\beta-\delta)T_2)\I-C)/T}\right)^{-1},
\end{equation}
and
\begin{equation}
\Pr(1\mid \rho_{\mathrm{in}},\beta,\delta)
=
\Tr\!\bigl[\widetilde M_{T,\delta,\beta}\rho_{\mathrm{in}}\bigr].
\end{equation}
In particular, if
\begin{equation}
\label{eq:offset-condition}
(\beta-\delta)T_2=\mu,
\end{equation}
then
\begin{equation}
\label{eq:betathermal}
\widetilde M_{T,\delta,\beta}
=
\left(\I+e^{(\mu\I-C)/T}\right)^{-1}
=
M_{T,\mu}.
\end{equation}
\end{proposition}

\begin{proof}
Let
\begin{equation}
C=\sum_{j=1}^d \lambda_j \ketbra{u_j}
\end{equation}
be the spectral decomposition of \(C\).
Under the evolution \(U_C=e^{-i\hat p\otimes C/T_2}\), each eigenmode \(\ket{u_j}\) shifts the
control position-quadrature distribution by \(\lambda_j/T_2\). Since the initial control density is
\(g_{T_1}(q-\delta)\), the position density conditioned on \(\ket{u_j}\) becomes
\begin{equation}
g_{T_1}\!\left(q-\delta-\frac{\lambda_j}{T_2}\right).
\end{equation}
Hence, the acceptance probability of mode \(j\) is
\begin{equation}
m_j
=
\int_{\beta}^{\infty}
g_{T_1}\!\left(q-\delta-\frac{\lambda_j}{T_2}\right)\,dq.
\end{equation}
Using the logistic cumulative distribution function
\begin{equation}
F_{T_1}(x)
\coloneqq
\int_{-\infty}^x g_{T_1}(q)\,dq
=
\frac{1}{1+e^{-x/T_1}},
\end{equation}
we obtain
\begin{align}
m_j
&=
1-
F_{T_1}\!\left(\beta-\delta-\frac{\lambda_j}{T_2}\right) \\
&=
\frac{1}{1+\exp\!\left(\frac{(\beta-\delta)T_2-\lambda_j}{T}\right)}.
\end{align}
Therefore
\begin{equation}
\widetilde M_{T,\delta,\beta}
=
\sum_{j=1}^d
\frac{1}{1+\exp\!\left(\frac{(\beta-\delta)T_2-\lambda_j}{T}\right)}
\ketbra{u_j}
=
\left(\I+e^{(((\beta-\delta)T_2)\I-C)/T}\right)^{-1},
\end{equation}
which proves Eq.~\eqref{eq:modified-effect}. The final statement follows immediately from
the condition~\eqref{eq:offset-condition}.
\end{proof}

The data-side quantum evolution Eq.~\eqref{eq:modified-Hevolution}
\begin{equation}
U_C = e^{-i\hat p\otimes C_\phi/T_2}
\end{equation} in Proposition~\ref{prop:modified-thermal measurement} no longer contains the large uninformative variable term \(d\mu\,\tau_{\mathrm{mm}}\). 
Meanwhile, the control qumode position $q$ has a logistic probability density
\begin{equation}
g_{T_1}(q-\delta)
=
\frac{e^{-(q-\delta)/T_1}}
{T_1\bigl(1+e^{-(q-\delta)/T_1}\bigr)^2}
=
\frac{1}{4T_1}
\sech^2\!\left(\frac{q-\delta}{2T_1}\right).
\end{equation}
Equivalently, the idealized position-space wavefunction is
\begin{equation}
\psi_{T_1,\delta}(q)
=
\sqrt{g_{T_1}(q-\delta)}
=
\frac{1}{2\sqrt{T_1}}
\sech\!\left(\frac{q-\delta}{2T_1}\right),
\end{equation}
so the control state is a smooth, symmetric, bell-shaped non-Gaussian wavepacket with
exponentially decaying tails that are broader than Gaussian tails. An approximate finite-energy realization of this shaped
continuous-variable wavepacket is sufficient in practice.

We now describe how to make this Hamiltonian simulation $U_C = e^{-i\hat p\otimes C_\phi/T_2}$.
Recall that Eq.~\eqref{eq:centered-covariance-main} gives the centered
feature-space covariance as
\begin{equation}
\label{eq:C-decomposition-final}
C_\phi=\bar\rho_\phi-\ketbra{m}.
\end{equation}
The empirical term is realized by the purification \(\ket{\Omega}\), while the
mean term is realized directly through the subnormalized mean branch
\begin{equation}
\label{eq:chi-m-branch}
\ket{\chi_m}
\coloneqq
(U_N^\dagger\otimes \I)\ket{\Omega}
=
\ket{0}\ket{m}+\ket{\perp},
\end{equation}
where the index component of \(\ket{\perp}\) is orthogonal to \(\ket{0}\).
We now instantiate the state-based Hamiltonian simulation using a deterministic
two substep product formula. Since the two terms in
Eq.~\eqref{eq:C-decomposition-final} have equal coefficient magnitude, there is
no need to randomly sample a label according to coefficient dependent
probabilities. Instead, each simulation round applies the empirical term and the
mean term once, with opposite signs.

Let
\begin{equation}
\Delta\coloneqq \frac{1}{R T_2},
\end{equation}
where \(R\) is the total number of deterministic simulation rounds. A single round consists of the following two substeps, 
as illustrated in Fig~\ref{fig:soft-pca-hamiltonian-sim}:
\begin{figure}[t]
\centering
\includegraphics[width=\textwidth]{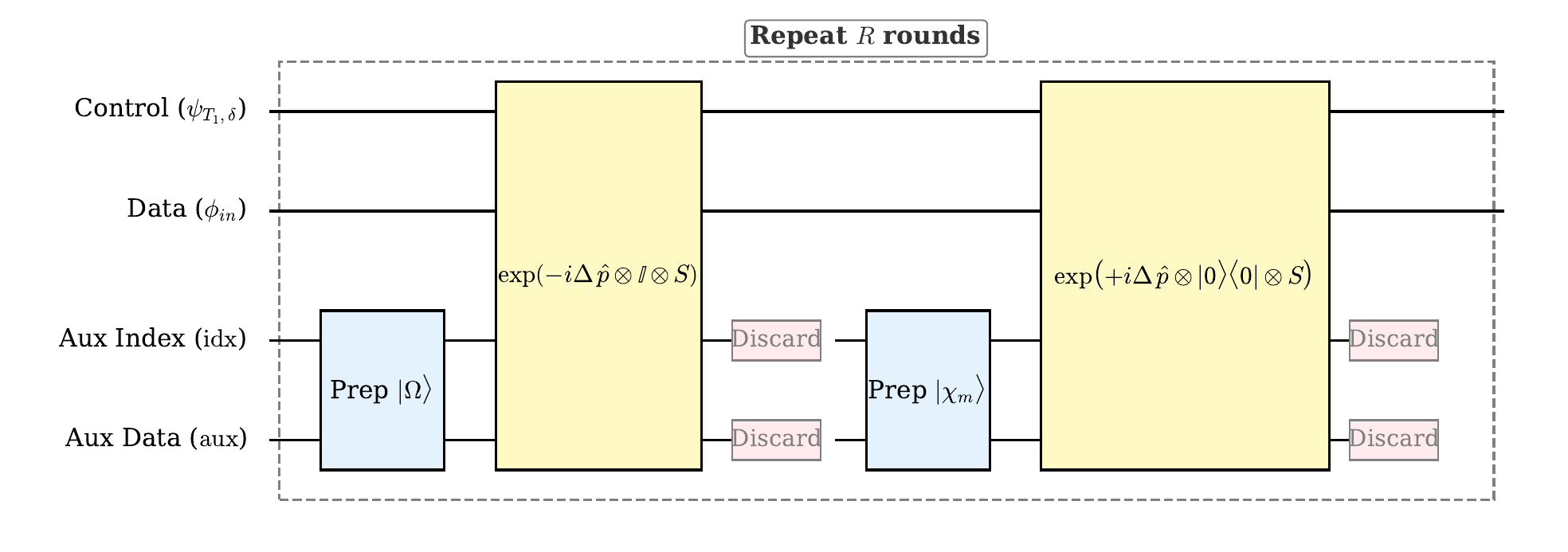}
\caption{\textbf{Implementation of the Hamiltonian evolution
\(
U_C\approx e^{-i\hat{p}\otimes C_\phi/T_2}
\)
for the soft PCA.} Each round applies two signed selector-swap substeps.
The empirical substep prepares
\(
\ket{\Omega}=N^{-1/2}\sum_i\ket{i}\ket{\phi_i}
\)
and applies
\(
\exp\!\left(
-i\,\Delta\,
\hat p\otimes
\I\otimes S_{\mathrm{data,aux}}\right)
\),
thereby generating \(+\bar\rho_\phi\). The mean-subtraction substep prepares
\(
\ket{\chi_m}=\ket0\ket m+\ket{\perp}
\)
and applies
\(
\exp\!\left(
+i\,\Delta\,
\hat p\otimes
\ketbra{0}\otimes S_{\mathrm{data,aux}}\right)
\),
thereby generating \(-\ketbra m\). The selector
\(\ketbra0_{\rm idx}\) is coherent, so no index postselection or
\(\alpha\)-estimation is required during Hamiltonian simulation. Repeating the
two substeps for \(R\) rounds yields Eq.~\eqref{eq:modified-Hevolution}.}
\label{fig:soft-pca-hamiltonian-sim}
\end{figure}

\begin{enumerate}[label=\arabic*.]
\item \textbf{Empirical covariance substep.}
Prepare one auxiliary copy of
\begin{equation}
\ket{\Omega}
=
\frac{1}{\sqrt N}\sum_{i=1}^N \ket{i}\ket{\phi_i}.
\end{equation}
Apply the short time swap evolution
\begin{equation}
\exp\!\left(
-i\,\Delta\,
\hat p\otimes
\I\otimes S_{\mathrm{data,aux}}\right),
\end{equation}
where \(S_{\mathrm{data,aux}}\) swaps the data register with the auxiliary data register and acts trivially on the index register. After discarding the
auxiliary registers, this substep implements the short time Hamiltonian term
\(+\bar\rho_\phi\).

\item \textbf{Mean-subtraction substep.}
Prepare one auxiliary copy of the subnormalized mean branch
\begin{equation}
\ket{\chi_m}
=
\ket{0}\ket{m}+\ket{\perp}.
\end{equation}
Apply the following operator which we call the \textit{opposite sign selector-swap evolution}
\begin{equation} \label{eq:selectorswap}
\exp\!\left(
+i\,\Delta\,
\hat p\otimes
\ketbra{0}\otimes S_{\mathrm{data,aux}}\right).
\end{equation}
The projector \(\ketbra{0}\) selects only the index \(\ket0\) branch
of Eq.~\eqref{eq:chi-m-branch}. Therefore, after discarding the auxiliary
registers, this substep implements the short time Hamiltonian term
\(-\ketbra m\).
\end{enumerate}

The selector in the second substep is coherent. We do not measure or postselect
the index register during the Hamiltonian simulation. The component
\(\ket{\perp}\) simply gives no contribution because it is annihilated by
\(\ketbra0_{\rm idx}\). Thus the same auxiliary index and feature registers can
be reset and reused in each round; no record of the index outcome, and no
additional register storing \(\alpha\) is required.

At the effective Hamiltonian level, the empirical substep contributes the
Hamiltonian
\(
+\hat p\otimes \bar\rho_\phi,
\)
while the mean-subtraction substep contributes
\(
-\hat p\otimes \ketbra m.
\)
Thus one deterministic round has the effective first order product form
\begin{equation}
U_{\rm round}
=
\exp\!\left(
+i\,\Delta\,\hat p\otimes \ketbra m
\right)
\exp\!\left(
-i\,\Delta\,\hat p\otimes \bar\rho_\phi
\right),
\qquad
\Delta=\frac{1}{R T_2}.
\end{equation}
By the first order Lie-Trotter formula~\cite{Trotter1959,Suzuki1990},
\begin{equation}
U_{\rm round}
=
\exp\!\left(
-i\,\Delta\,\hat p\otimes
(\bar\rho_\phi-\ketbra m)
\right)
+
O(\Delta^2).
\end{equation}
Since
\begin{equation}
C_\phi=\bar\rho_\phi-\ketbra m,
\end{equation}
repeating this deterministic round \(R\) times gives
\begin{equation}
U_{\rm round}^R
\approx
\exp\!\left(
-i\,\frac{1}{T_2}\hat p\otimes C_\phi
\right)
=
U_C.
\end{equation}

As a direct specialization of the swap based state exponentiation analysis of
~\cite{LiuWilde2026}, for any simulation error parameter
\(\varepsilon_{\mathrm{sim}}>0\), it suffices to take
\begin{equation}
\label{eq:estimationcost}
R
=
O\!\left(
\frac{1}{T_2^2\,\varepsilon_{\mathrm{sim}}}
\right)
\end{equation}
deterministic simulation rounds, where each round contains two controlled selector-swap substeps.

We have thus provided the complete implementation of the
modified thermal measurement described in
Proposition~\ref{prop:modified-thermal measurement}. We next show that this primitive suffices to calibrate the optimal soft PCA filters.

\subsection{Training-free soft PCA filter}
With the control-side displacement fixed to $\delta=0$, the parameter $\mu$ is no longer embedded in the data-side Hamiltonian. Instead, it is absorbed directly into the classical post-measurement threshold.
Following Proposition~\ref{prop:modified-thermal measurement}, we prepare the control qumode state
\begin{equation}
\ket{\psi_{T_1}}=
\ket{\psi_{T_1,0}}
=
\int_{\mathbb R}\sqrt{g_{T_1}(q)}\,\ket q\,dq,
\qquad
g_{T_1}(q)
=
\frac{e^{-q/T_1}}{T_1(1+e^{-q/T_1})^2}.
\end{equation}
With this fixed control state, all filters in the family
\begin{equation}
M_{T,\mu}=\left(\I+e^{(\mu\I-C)/T}\right)^{-1}
\end{equation}
are generated by the same fixed quantum evolution
\begin{equation}
\label{eq:uclast}
U_C=e^{-i\hat p\otimes C/T_2},
\end{equation}
and differ only by the classical threshold
\begin{equation}
\beta=\mu/T_2.
\end{equation}

For any probe state \(\rho_{\rm in}\), the circuit
\begin{equation}
\ket{\psi_{T_1}}\otimes \rho_{\rm in}
\ \xrightarrow{\;U_C=e^{-i\hat p\otimes C/T_2}\;}
\ \text{measure }q
\end{equation}
followed by the classical rule \(q>\beta\) implements
\begin{equation}
\Pr(q>\beta\mid \rho_{\rm in})
=
\Tr\!\left(M_{T,\beta T_2}\rho_{\rm in}\right).
\end{equation}
Thus the circuit can be used to estimate
\(\Tr(M_{T,\beta T_2}\rho_{\rm in})\) for any input state \(\rho_{\rm in}\). When the probe state is chosen as \(\rho_{\mathrm{in}}=\tau_{\mathrm{mm}}=\I/d\) for rank calibration, the
entire one-parameter family \(\{\Tr(M_{T,\mu})\}_{\mu\in\mathbb R}\) is encoded in the single position
distribution produced by this fixed quantum circuit.
A lightweight schematic of this fixed measurement primitive is shown in Fig.~\ref{fig:schematic}. We now delve into the details.
\begin{figure}[t]
    \centering
    \includegraphics[width=0.85\textwidth]{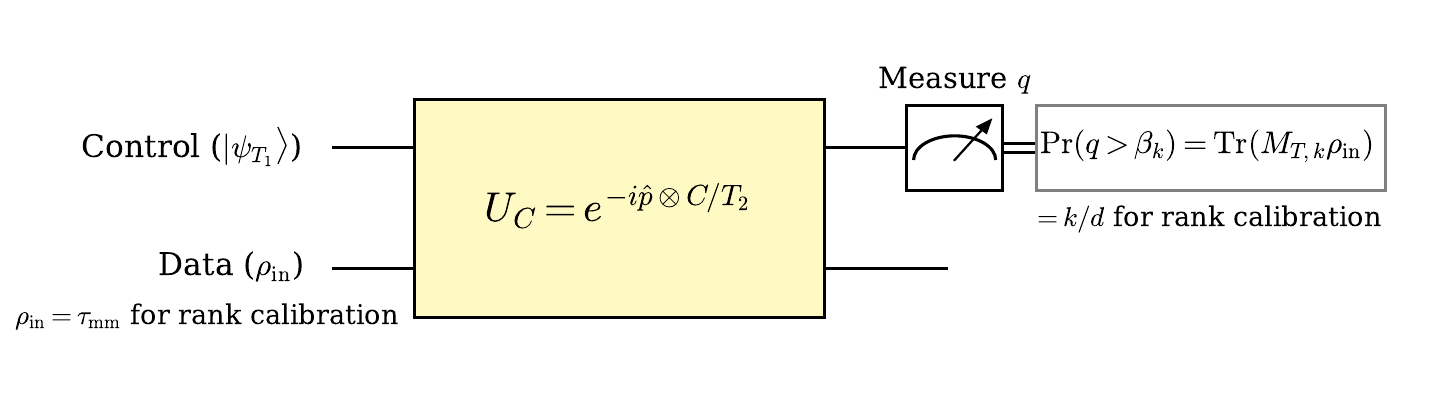}
    \caption{\textbf{Schematic of the quantum soft PCA.} A control qumode prepared in $\ket{\psi_{T_1}}$ and a data register in a generic probe state $\rho_{\mathrm{in}}$ undergo the joint Hamiltonian evolution $U_C = e^{-i\hat{p}\otimes C/T_2}$. By measuring the position quadrature $q$ of the control qumode and applying a simple classical threshold rule $q > \beta$, the protocol directly evaluates the soft PCA score $\Tr(M_{T,\beta T_2}\rho_{\mathrm{in}})$. This unified schematic applies to both the trace-budget calibration phase (e.g. $\rho_{\mathrm{in}} = \mathbb{I}/d$ for rank calibration) and the downstream inference on new test inputs. A specific implementation of the quantum circuit is shown in Fig~\ref{fig:soft-pca-hamiltonian-sim}, and the detailed rank calibration process is illustrated in Fig.~\ref{fig:training-free-soft-pca}.}
    \label{fig:schematic}
\end{figure}

In the rank calibration case, if
\(C=\sum_j \lambda_j\ketbra{u_j}\), the output position density is
\begin{equation}
f_C(q)
=
\frac1d\sum_{j=1}^d
g_{T_1}\!\left(q-\frac{\lambda_j}{T_2}\right).
\end{equation}
Therefore, the peaks in the physical position distribution are centered at \(\lambda_j/T_2\), while their width is determined by the control input distribution \(g_{T_1}\), equivalently by the parameter \(T_1\). In eigenvalue
units the corresponding smoothing scale is \(T=T_1T_2\).

Moreover, the optimal soft PCA filter for
any legal trace budget \(k\in\{1,\dots,d-1\}\) can be obtained directly by a classical quantile-selection post-processing step, rather than by an iterative training.
Let \(F_C\) denote the cumulative distribution function of the position outcome produced by the fixed circuit
Eq.~\eqref{eq:uclast}
with probe state \(\tau_{\mathrm{mm}}\), namely
\begin{equation}
F_C(\beta)\coloneqq \Pr(q\le \beta \mid \tau_{\mathrm{mm}}).
\end{equation}
Then, by Proposition~\ref{prop:modified-thermal measurement}, for every \(\beta\in\mathbb R\),
\begin{equation}
\label{eq:tail-trace-identity-short}
1-F_C(\beta)
=\Tr\bigl(M_{T,\beta T_2}\tau_{\mathrm{mm}}\bigr)
=\frac{1}{d}\Tr\bigl(M_{T,\beta T_2}\bigr).
\end{equation}
We get the following training-free optimal soft PCA filter as a calibration for all possible $k$
which is demonstrated in Fig.~\ref{fig:training-free-soft-pca}:

\begin{figure}[t]
    \centering
    \includegraphics[width=0.7\textwidth]{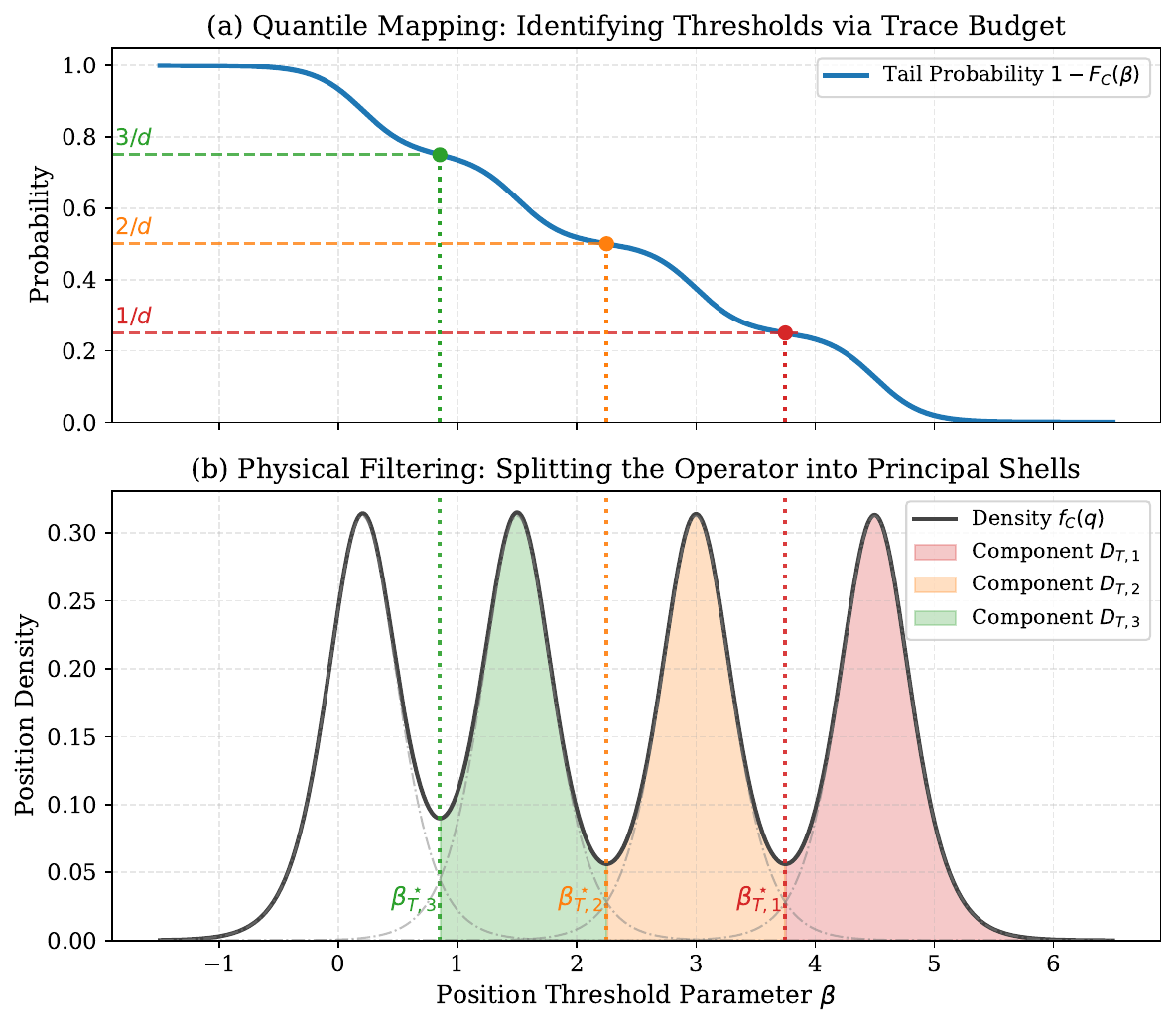}
    \caption{\textbf{Training-free calibration of optimal soft PCA filters via quantile mapping. }(a) The tail probability $1-F_C(\beta)$ decreases in a staircase-like manner. The optimal thresholds $\beta^\star_{k}$ are found strictly at the quantiles $k/d$. (b) The corresponding physical position density $f_C(q)$ (solid line), composed of distinct eigenmode peaks (dashed gray lines). The thresholds vertically partition the measured distribution into adjacent quantile windows,
    which implement the soft spectral components \(D_{T,k}\). When the relevant eigengaps are large and the peaks are well separated, the thresholds tend to lie in low-density valleys, making the selected filter robust to small threshold perturbations.}
    \label{fig:training-free-soft-pca}
\end{figure}

\begin{theorem}[Training-free soft PCA filter calibration]
\label{thm:quantile-soft-pca}
Fix \(\delta=0\) in Proposition~\ref{prop:modified-thermal measurement}. For each legal trace budget
\(k\in\{1,\dots,d-1\}\), let \begin{equation}
M_{T,k}^\star
\coloneqq
\left(\I+e^{(\mu_{k}^\star\I-C)/T}\right)^{-1}
\end{equation}
be the unique soft PCA optimizer characterized by
$\Tr(M_{T,k}^\star)=k$. Then the corresponding threshold
$\beta_k^\star\coloneqq \mu_k^\star/T_2$
is the unique solution of
\begin{equation}
\label{eq:quantile-equation}
1-F_C(\beta)=\frac{k}{d},
\end{equation}
and
\begin{equation}
M_{T,k}^\star=\widetilde M_{T,\beta_k^\star}.
\end{equation}
Hence the optimal soft PCA filter for every target rank \(k\) is obtained by classical quantile selection from the same fixed circuit position distribution, without any iterative quantum training or rank-dependent quantum circuit update.

Moreover, let \(q_1,\dots,q_S\) be \(S\) independent position samples from the same fixed circuit, and
\begin{equation}\widehat F_S(\beta)\coloneqq \frac{1}{S}\sum_{s=1}^S \mathbf 1\{q_s\le \beta\}\end{equation}
be the empirical cumulative distribution function. To guarantee that the empirical soft filter $\widehat{M}_{T,k}^\star$ satisfies a trace error bound
\begin{equation}
    \sup_{k\in\{1,\dots,d-1\}} \left| \Tr(\widehat{M}_{T,k}^\star) - k \right| \le \eta
\end{equation}
with probability at least $1-\delta_{\mathrm{fail}}$, it suffices to require a CDF estimation accuracy of $\eta/d$ and then requires a sample complexity of
\begin{equation}
S\ge \frac{d^2}{2\eta^2}\log\frac{2}{\delta_{\mathrm{fail}}},
\end{equation}
\end{theorem}

\begin{proof}
For each \(k\), let \(\mu_k^\star\) be the unique solution of
\begin{equation}
\Tr(M_{T,\mu})=k.
\end{equation}
Setting
\begin{equation}
\beta_k^\star=\mu_k^\star/T_2,
\end{equation}
Eq.~\eqref{eq:tail-trace-identity-short} gives
\begin{equation}
1-F_C(\beta_k^\star)
=
\frac{1}{d}\Tr(M_{T,\mu_k^\star})
=
\frac{k}{d},
\end{equation}
so \(\beta_k^\star\) satisfies Eq.~\eqref{eq:quantile-equation}. Proposition~\ref{prop:modified-thermal measurement} then implies
\begin{equation}
\widetilde M_{T,\beta_k^\star}=M_{T,\mu_k^\star}=M_{T,k}^\star.
\end{equation}
Uniqueness follows from the strict monotonicity of \(\mu\mapsto \Tr(M_{T,\mu})\).

For the empirical statement, apply the Dvoretzky-Kiefer-Wolfowitz inequality with Massart's sharp constant~\cite{Dvoretzky1956,Massart1990} to the empirical CDF \(\widehat F_S\):
\begin{equation}
\label{eq:all-k-dkw-samples}
\Pr\!\left(
\sup_{\beta\in\mathbb R}
|\widehat F_S(\beta)-F_C(\beta)|>\eta/d
\right)
\le 2e^{-2S\eta^2/d^2}.
\end{equation}
This yields the stated sample complexity bound.
\end{proof}

The rank calibration procedure corresponding to
Theorem~\ref{thm:quantile-soft-pca} is summarized in
Algorithm~\ref{alg:rank-calibration}. It outputs a collection of classical thresholds. Once a threshold \(\beta_k\) is stored, the
corresponding soft filter is deployed by rerunning the same fixed thermal
measurement circuit and accepting the event \(q>\beta_k\). Thus all
rank dependence is confined to classical post-processing of the measured position variable.

\begin{algorithm}[t]
\caption{Rank calibration by empirical quantiles}
\label{alg:rank-calibration}
\begin{algorithmic}[1]
\Require Fixed circuit \(U_C=e^{-i\hat p\otimes C/T_2}\), maximally mixed probe
\(\tau_{\mathrm{mm}}=\I/d\), number of samples \(S\), and requested ranks
\(k\in\{1,\ldots,d-1\}\).

\State Run the fixed circuit $U_C$ with probe \(\tau_{\mathrm{mm}}\) for \(S\) shots
and record the position outcomes \(q_1,\ldots,q_S\).
\State Form the empirical CDF
\[
\widehat F_S(\beta)
\coloneqq
\frac1S\sum_{s=1}^S \mathbf 1\{q_s\le \beta\}.
\]
\For{each requested rank \(k\)}
    \State Set
    \[
    \widehat\beta_k
    \coloneqq
    \inf\left\{
    \beta\in\mathbb R:
    \widehat F_S(\beta)\ge 1-\frac{k}{d}
    \right\}.
    \]
\EndFor
\State \Return Calibrated thresholds \(\{\widehat\beta_k\}\).
\end{algorithmic}
\end{algorithm}

\begin{remark}
    Here the protocol for the training-free soft PCA filter is in fact very closely related to the Power of One Qumode model~\cite{liu2016power}, where the Hamiltonian in the Power of Qumode model is here replaced by that of the covariance matrix $C$. 
    When using the input state $\I/d$ and ancilla squeezed state input, one is also able to recover the eigenvalues of $C$ by studying the location of the minima of the probability distribution of the output position measurements, with eigenvalue resolution determined by the amount of squeezing in the input qumode ancilla state. 
\end{remark}

\begin{proposition}[Task-dependent sampling for a fixed target rank]
\label{prop:fixed-rank-sampling}
The all-\(k\) guarantee in Theorem~\ref{thm:quantile-soft-pca} is a uniform CDF guarantee. If, instead, one targets a single fixed rank \(K\), the relevant tail probability is
\begin{equation}
r_K\coloneqq \frac{K}{d}.
\end{equation}
To obtain an \(O(1)\) accuracy in the trace budget, one must estimate this tail probability to accuracy \(O(1/d)\). A Bernstein-type binomial concentration bound~\cite{Boucheron2013} gives the following sharper sufficient sample complexity: for trace accuracy \(\eta>0\) and failure probability \(\delta_{\mathrm{fail}}\), it is sufficient to take
\begin{equation}
\label{eq:fixed-K-sampling}
S
\ge
\left(
\frac{2K(d-K)}{\eta^2}
+
\frac{2d}{3\eta}
\right)
\log\frac{2}{\delta_{\mathrm{fail}}}.
\end{equation}
Thus, up to constants and logarithmic factors,
\begin{equation}
S
=
O\!\left(
\frac{K(d-K)}{\eta^2}
+
\frac{d}{\eta}
\right).
\end{equation}
In particular:
\begin{enumerate}[label=\roman*.]
    \item if \(K=O(1)\), then \(S=\widetilde{O}(d)\);
    \item if \(K\approx d/2\), then \(S=\widetilde{O}(d^2)\), matching the worst case all-\(k\) scaling;
    \item if only the normalized rank fraction \(K/d\) is required to constant accuracy, then the required probability accuracy is \(O(1)\), and the sampling cost is \(O(1)\).
\end{enumerate}
\end{proposition}
\begin{proof}
For a fixed threshold \(\beta\), the event \(\{q>\beta\}\) is Bernoulli with success probability
\begin{equation}
G_C(\beta)\coloneqq 1-F_C(\beta).
\end{equation}
At the population rank-\(K\) threshold, \(G_C(\beta_K^\star)=K/d=r_K\). A trace error \(\eta\) corresponds to a tail-probability error
\begin{equation}
\epsilon_q=\frac{\eta}{d}.
\end{equation}
For Bernoulli samples with variance \(r_K(1-r_K)=K(d-K)/d^2\), Bernstein's inequality gives
\begin{equation}
\Pr\!\left(
|\widehat G_S-r_K|>\epsilon_q
\right)
\le
2\exp\!\left(
-
\frac{S\epsilon_q^2}
{2r_K(1-r_K)+2\epsilon_q/3}
\right).
\end{equation}
Substituting \(r_K=K/d\) and \(\epsilon_q=\eta/d\) yields Eq.~\eqref{eq:fixed-K-sampling}. 
\end{proof}

The distinction between Theorem~\ref{thm:quantile-soft-pca} and Proposition~\ref{prop:fixed-rank-sampling} is purely statistical and task dependent. The quantum circuit is unchanged. The \(d^2\) factor in the all-\(k\) bound is not caused by the Fermi--Dirac filter itself, but by the normalization
\begin{equation}
\Pr(q>\beta\mid \tau_{\mathrm{mm}})
=
\frac{1}{d}\Tr(M_{T,\beta T_2}).
\end{equation}
Thus a constant additive error in the unnormalized trace requires an \(O(1/d)\) additive error in a probability.

While Theorem~\ref{thm:quantile-soft-pca} controls the uniform CDF error, bounding the Euclidean distance of the threshold parameter $|\widehat\beta_k-\beta_k^\star|$ is both mathematically ill-conditioned and physically unnecessary. Operationally, the primary object of interest is not the coordinate $\beta$, but the selected soft filter $\widehat{M}_{T,k} \coloneqq M_{T,\widehat\beta_k T_2}$.
Since all filters \(M_{T,\beta T_2}\) are functions of the same operator \(C\), any two such filters commute and are ordered by the scalar threshold \(\beta\). Hence, for two thresholds \(\beta,\beta'\), the difference \(M_{T,\beta T_2}-M_{T,\beta'T_2}\) has a fixed sign and
\begin{equation}
\|M_{T,\beta T_2}-M_{T,\beta'T_2}\|_1
=
\left|\Tr(M_{T,\beta T_2})-\Tr(M_{T,\beta'T_2})\right|.
\end{equation}
Therefore, on the event \(\sup_\beta|\widehat F_S(\beta)-F_C(\beta)|\le \eta/d\), the empirically selected filter satisfies
\begin{equation}
\|\widehat M_{T,k}-M_{T,k}^\star\|_1\le \eta
\end{equation}

Consequently, the error in any downstream unnormalized score $\Tr(\widehat{M}_{T,k}\rho_{\rm in})$ is bounded by $\eta$, completely independent of the numerical fluctuation of $\widehat\beta_k$.

This distinction is crucial in the highly favorable large eigengap regime. A large spectral gap $\lambda_k - \lambda_{k+1}$ creates a near-zero ``probability valley" in the position density $f_C(q)$. While this flat region causes the empirical coordinate $\widehat\beta_k$ to be statistically ill-conditioned and exhibit high variance, any threshold drifting within this spectral gap does not cross new eigenvalue shells. This explains why operator-level calibration can remain stable even when the numerical threshold coordinate itself has high variance.

\begin{corollary}[Task-dependent quantum sampling complexity]
\label{cor:gate-complexity}
Each position sample from the fixed circuit \(U_C\) requires
\begin{equation}
O\!\left(\frac{1}{T_2^2\,\varepsilon_{\mathrm{sim}}}\right)
\end{equation}
controlled selector-swap substeps. Therefore:

\begin{enumerate}[label=\roman*.]
    \item \textbf{All-\(k\) trace calibration.}
    To determine the entire nested rank ladder with uniform trace accuracy \(\eta\) and failure probability at most \(\delta_{\mathrm{fail}}\), it suffices to use
    \begin{equation}
    O\!\left(
    \frac{d^2}{T_2^2\,\varepsilon_{\mathrm{sim}}\,\eta^2}
    \log\frac{2}{\delta_{\mathrm{fail}}}
    \right)
    \end{equation}
    controlled selector-swap substeps (see Eq.~\eqref{eq:selectorswap}).

    \item \textbf{Single fixed-rank calibration.}
    For one fixed target rank \(K\), the sharper binomial bound gives
    \begin{equation}
    O\!\left(
    \frac{1}{T_2^2\,\varepsilon_{\mathrm{sim}}}
    \left[
    \frac{K(d-K)}{\eta^2}
    +
    \frac{d}{\eta}
    \right]
    \log\frac{2}{\delta_{\mathrm{fail}}}
    \right)
    \end{equation}
    controlled selector-swap substeps.

    \item \textbf{Fractional rank calibration.}
    If the task only requires the normalized trace fraction \(K/d\) to additive accuracy \(\eta\), then the required probability accuracy is \(O(\eta)\), and the sampling contribution becomes
    \begin{equation}
    O\!\left(
    \frac{1}{T_2^2\,\varepsilon_{\mathrm{sim}}\,\eta^2}
    \log\frac{2}{\delta_{\mathrm{fail}}}
    \right).
    \end{equation}
\end{enumerate}
Multiplying these bounds by the gate cost per selector-swap micro-step gives the corresponding gate complexity.
\end{corollary}

\begin{proof}
The first item combines the DKW sample complexity in Eq.~\eqref{eq:all-k-dkw-samples} with the cost of one simulated thermal measurement sample. The second item combines Proposition~\ref{prop:fixed-rank-sampling} with the same per-sample simulation cost. The third item is the standard Bernoulli/Hoeffding scaling for estimating a probability to constant additive accuracy.
\end{proof}

Although Theorem~\ref{thm:quantile-soft-pca} provides the thresholds for all
\(k\) simultaneously, quantitative approximation to the hard rank-\(k\) spectrum is a
\emph{per-\(k\)} statement. For each \(k\), the entropy argument gives
\begin{equation}
0\le v_{\mathrm{hard}}^{(k)}-\Tr(CM_{T,k}^\star)\le T\,d\,h(k/d),
\end{equation}
while under the low temperature local eigengap condition $\lambda_k(C)>\mu_k^\star>\lambda_{k+1}(C)$,
\begin{equation}
\label{eq:per-k-gap-bound}
\varepsilon_{T,k}\coloneqq \|M_{T,k}^\star-P_k\|_{\mathrm{op}}
\le e^{-g_{T,k}/T},
\qquad
g_{T,k}\coloneqq
\min\{\lambda_k(C)-\mu_k^\star,\ \mu_k^\star-\lambda_{k+1}(C)\}.
\end{equation}

Therefore, a temperature chosen to give a prescribed accuracy for one target rank \(K\) does not
automatically provide the same guarantee for all larger or smaller \(k\). The all-\(k\) profile is
always available, but only those ranks for which both \(T\,d\,h(k/d)\) and
\(e^{-g_{T,k}/T}\) are sufficiently small admit a certified approximation to the corresponding hard
PCA spectrum.

The thresholds \(\beta_k^\star\) are strictly decreasing in \(k\). Hence the
filters are nested:
\begin{equation}
0\preceq M_{T,1}^\star \preceq M_{T,2}^\star \preceq
\cdots \preceq M_{T,d-1}^\star \preceq \I .
\end{equation}
Define
\begin{equation}
M_{T,0}^\star\coloneqq 0,
\qquad
M_{T,d}^\star\coloneqq \I,
\qquad
D_{T,k}\coloneqq M_{T,k}^\star-M_{T,k-1}^\star,
\quad k=1,\dots,d.
\end{equation}
Then
\begin{equation}
D_{T,k}\succeq 0,
\qquad
\Tr(D_{T,k})=1,
\qquad
\sum_{k=1}^d D_{T,k}=\I.
\end{equation}
Thus \(\{D_{T,k}\}_{k=1}^d\) forms a soft spectral resolution of the covariance
operator \(C\). At low temperature, each \(D_{T,k}\) concentrates around the
corresponding principal component.

Operationally, the operators \(D_{T,k}\) are obtained by resolving the position
outcome into adjacent quantile bins. Let
\begin{equation}
W_{1}\coloneqq (\beta_{1}^\star,\infty),
\end{equation}
\begin{equation}
W_{k}\coloneqq
(\beta_k^\star,\beta_{k-1}^\star],
\qquad
k=2,\dots,d-1,
\end{equation}
and
\begin{equation}
W_{d}\coloneqq (-\infty,\beta_{d-1}^\star].
\end{equation}
Then the window event \(q\in W_{k}\) implements the POVM effect \(D_{T,k}\).

\begin{proposition}[Soft filters and component states from window-conditioned outputs]
\label{prop:shell-output-Dk}
Let the probe state be the maximally mixed state $\tau_{\mathrm{mm}} = \I/d$. For a measurable position window $W \subseteq \mathbb{R}$, let $E(W)$ be its associated POVM effect. Whenever $\Tr(E(W)) > 0$, the postselected data-register state conditioned on $q \in W$ is
\begin{equation}
\label{eq:window-conditioned-output}
\rho_{\mathrm{out}\mid W} = \frac{E(W)}{\Tr(E(W))}.
\end{equation}

In particular, this framework prepares the following states:
\begin{enumerate}
    \item For the position window $W_k$, the associated POVM effect is $E(W_k) = D_{T,k}$, and the bin-conditioned output state is exactly the soft component state:
    \begin{equation}
    \rho_{\mathrm{out}\mid W_k} = D_{T,k}.
    \end{equation}

    \item For the cumulative tail window $W_{\le K} \coloneqq (\beta_K^\star, \infty)$, the associated POVM effect is $E(W_{\le K}) = M_{T,K}^\star$, and the tail-conditioned output state is the normalized cumulative soft filter:
    \begin{equation}
    \rho_{\mathrm{out}\mid W_{\le K}} = \frac{1}{K}M_{T,K}^\star.
    \end{equation}
\end{enumerate}
\end{proposition}

\begin{proof}
Recall that the measurement instrument corresponding to a window $W$ is defined as $\mathcal I_W(\rho) \coloneqq \int_W K_q\rho K_q^\dagger\,dq$, where the Kraus operators are
\begin{equation}
K_q = \bra q U_C\ket{\psi_{T_1}} = \sum_{i=1}^d \sqrt{g_{T_1}\!\left(q-\frac{\lambda_i}{T_2}\right)} \,\ketbra{u_i}.
\end{equation}
Since $K_q$ is Hermitian and diagonal in the eigenbasis of $C$, the POVM effect associated with $W$ is exactly $E(W) = \int_W K_q^\dagger K_q\,dq = \int_W K_q^2\,dq$.

Evaluating the instrument on the maximally mixed state $\tau_{\mathrm{mm}} = \I/d$, we find
\begin{equation}
\mathcal I_W(\tau_{\mathrm{mm}}) = \int_W K_q\frac{\I}{d}K_q\,dq = \frac{1}{d} \int_W K_q^2\,dq = \frac{1}{d}E(W).
\end{equation}
The probability of obtaining a measurement outcome $q \in W$ is therefore
\begin{equation}
\Pr(q\in W\mid \tau_{\mathrm{mm}}) = \Tr\!\left[\mathcal I_W(\tau_{\mathrm{mm}})\right] = \frac{1}{d}\Tr(E(W)).
\end{equation}
Normalizing the postselected state yields the general relation in Eq.~\eqref{eq:window-conditioned-output}:
\begin{equation}
\rho_{\mathrm{out}\mid W} = \frac{\mathcal I_W(\tau_{\mathrm{mm}})}{\Tr[\mathcal I_W(\tau_{\mathrm{mm}})]} = \frac{E(W)}{\Tr(E(W))}.
\end{equation}

Applying this general relation to the specific windows yields the remaining claims:
For the position window $W_k$, we have $E(W_k) = D_{T,k}$ by definition. Since $\Tr(D_{T,k})=1$, the success probability is $1/d$, and the conditioned state is $\rho_{\mathrm{out}\mid W_k} = D_{T,k}$.

For the cumulative tail window $W_{\le K}$, the POVM effect is the cumulative soft filter $E(W_{\le K}) = M_{T,K}^\star$. Since $\Tr(M_{T,K}^\star)=K$, the success probability is $K/d$, yielding the conditioned state $\rho_{\mathrm{out}\mid W_{\le K}} = M_{T,K}^\star / K$.
\end{proof}

The operator \(D_{T,k}\) should be interpreted as the \(k\)-th soft principal
component state. At finite temperature it need not be rank one. In the nondegenerate low temperature regime, \(D_{T,k}\) concentrates around the
usual \(k\)-th principal eigenspace. More precisely, assume
\begin{equation}
\lambda_1>\lambda_2>\cdots>\lambda_{K+1}.
\end{equation}
For \(k=1,\dots,K\), let
\begin{equation}
\Delta_{T,k}
\coloneqq
\varepsilon_{T,k}+\varepsilon_{T,k-1},
\qquad
\varepsilon_{T,0}\coloneqq 0.
\end{equation}
Then
\begin{equation}
\label{eq:Dk-rank-one-approx}
\|D_{T,k}-\ketbra{u_k}\|_{\mathrm{op}}
\le
\Delta_{T,k}.
\end{equation}
Thus, as \(T\to 0\) under a fixed eigengap, the window-conditioned component
state \(D_{T,k}\) converges in operator norm to the rank-one principal projector
\(\ketbra{u_k}\).

Consequently, for any normalized input state \(\rho_{\mathrm{in}}\),
\begin{equation}
\left|
\Tr(D_{T,k}\rho_{\mathrm{in}})
-
\Tr(\ketbra{u_k}\rho_{\mathrm{in}})
\right|
\le
\Delta_{T,k}.
\end{equation}

Thus, if a classical description is needed, it can be obtained by tomography of
the window-conditioned output states \(D_{T,k}\) and \(M_{T,K}^\star/K\),
respectively.

\subsection{Normalized retained variance calibration}
\label{sec:covariance-probe-profile}

The maximally mixed probe used in Theorem~\ref{thm:quantile-soft-pca}
calibrates the trace of the soft filter. A second natural probe gives the
retained variance calibration directly. This is closer to the classical explained-variance criterion often used for selecting the number of principal components. A covariance-weighted probe would estimate the retained variance fraction directly, so this formulation avoids the \(1/d\) trace-conversion factor that appears in rank calibration. Let
\begin{equation}
V_C\coloneqq \Tr(C),
\qquad
\rho_C\coloneqq \frac{C}{V_C}.
\end{equation}
We get the normalized covariance probe \(\rho_C\) using the LCU centered-input preparation described in Section~\ref{sec:qcovariance}. 
Indeed, sampling \(i\) uniformly and postselecting the centered branch
\begin{equation}
\ket{z_i}\coloneqq \ket{\phi_i}-\ket m
\end{equation}
produces, conditioned on success, the mixed state
\begin{equation}
\label{eq:rho_c}
\frac{1}{V_C}
\frac1N\sum_{i=1}^N \ketbra{z_i}
=
\frac{C_\phi}{V_C}
=
\rho_C.
\end{equation}
Because the procedure conditions on the successful centered LCU branch, \(V_C\) need not be known classically.

Run the same quantum evolution circuit
\begin{equation}
U_C=e^{-i\hat p\otimes C/T_2}
\end{equation}
with input probe \(\rho_C\), and define the tail function
\begin{equation}
G_C(\beta)
\coloneqq
\Pr(q>\beta\mid \rho_C).
\end{equation}
By Proposition~\ref{prop:modified-thermal measurement},
\begin{equation}
\label{eq:covariance-probe-tail-main}
G_C(\beta)
=
\Tr\!\left(\rho_C M_{T,\beta T_2}\right).
\end{equation}
Thus \(G_C(\beta)\) is exactly the normalized retained covariance energy of the
soft filter selected by threshold \(\beta\).
\begin{remark}
For the unnormalized target $\Tr\!\left(C M_{T,\beta T_2}\right)$, recall that the centered LCU preparation success probability is
\begin{equation}
P_{\rm succ}
=
\frac14\,\Tr(C)
=
\frac{V_C}{4},
\end{equation}
and the conditioned state is \(\rho_C=C/V_C\). Hence the joint probability
\begin{equation}
\Pr(q>\beta)
=
\frac{V_C}{4}\,
\Tr(\rho_C M_{T,\beta T_2})
=
\frac14\Tr(C M_{T,\beta T_2}).
\end{equation}
This can also be used as a calibration target if needed. 
\end{remark}

In parallel, the maximally mixed probe gives
\begin{equation}
\Pr(q>\beta\mid \tau_{\rm mm})
=
\frac{1}{d}\Tr(M_{T,\beta T_2}).
\end{equation}
Together, the maximally mixed scan and covariance-probe scan produce the two quantities needed to describe the soft explained-variance profile:
\begin{equation}
\beta
\longmapsto
\left(
\Tr(M_{T,\beta T_2}),
\;
\Tr(\rho_C M_{T,\beta T_2})
\right).
\end{equation}

For a target normalized variance level \(\vartheta\in(0,1)\), the threshold
\(\beta_\vartheta\) is determined by
\begin{equation}
G_C(\beta_\vartheta)=\vartheta,
\end{equation}
and the associated variance-indexed soft filter is
\begin{equation}
M_{T,\vartheta}
\coloneqq
M_{T,\beta_\vartheta T_2}.
\end{equation}
This filter is also the trace-constrained soft PCA optimizer for the soft trace
budget
\(
k_T(\vartheta)
=
\Tr(M_{T,\vartheta}).
\)
Hence the covariance-probe scan does not introduce a different finite-temperature
optimization problem; it selects a point on the same trace-constrained soft PCA
solution path, indexed by retained variance rather than by trace.

The corresponding empirical retained-variance calibration is given in
Algorithm~\ref{alg:retained-variance-calibration}. As in the rank case, the
algorithm only returns thresholds; deployment again uses the same fixed circuit
with the tail rule \(q>\beta_\vartheta\).

\begin{algorithm}[t]
\caption{Retained-variance calibration by empirical tail probabilities}
\label{alg:retained-variance-calibration}
\begin{algorithmic}[1]
\Require Fixed circuit \(U_C=e^{-i\hat p\otimes C/T_2}\), covariance probe
\(\rho_C=C/V_C\), number of successful probe samples \(S\), and requested
levels \(\vartheta\in(0,1)\).

\State Run the fixed circuit $U_C$ with probe \(\rho_C\) for \(S\) successful shots
and record the position outcomes \(q_1,\ldots,q_S\).
\State Form the empirical tail function
\[
\widehat G_C(\beta)
\coloneqq
\frac1S\sum_{s=1}^S \mathbf 1\{q_s>\beta\}.
\]
\For{each requested retained-variance level \(\vartheta\)}
    \State Set
    \[
    \widehat\beta_\vartheta
    \coloneqq
    \inf\left\{
    \beta\in\mathbb R:
    \widehat G_C(\beta)\le \vartheta
    \right\}.
    \]
\EndFor
\State \Return Calibrated thresholds \(\{\widehat\beta_\vartheta\}\).
\end{algorithmic}
\end{algorithm}

\begin{proposition}[Sample complexity of the retained variance scan]
\label{prop:cov-probe-profile-complexity}
Let \(q_1,\dots,q_S\) be independent position samples obtained from the fixed
circuit with probe \(\rho_C\), and let \(\widehat G_C\) be the empirical tail
function. To guarantee
\begin{equation}
\sup_{\beta\in\mathbb R}
\left|
\widehat G_C(\beta)-G_C(\beta)
\right|
\le
\eta
\end{equation}
with probability at least \(1-\delta\), it suffices to take
\begin{equation}
S
\ge
\frac{1}{2\eta^2}
\log\frac{2}{\delta}.
\end{equation}
Consequently, the entire retained variance profile, and in particular the
threshold for any target \(\vartheta\in(0,1)\), is calibrated with
dimension-independent successful-sample complexity
\(O(\eta^{-2})\). If raw centered-LCU preparation attempts are counted, the expected number of
attempts is larger by the factor \(4/V_C\).
\end{proposition}

\begin{proof}
The claim follows directly from the Dvoretzky--Kiefer--Wolfowitz inequality
applied to the position samples from the covariance-probe distribution. Since
the retained variance tail is a Born probability, no \(1/d\) trace conversion is
needed.
\end{proof}

\begin{proposition}[Low-temperature explained-variance limit]
\label{prop:low-temp-explained-variance}
Let
\begin{equation}
C_\phi
=
\sum_{j=1}^{d}
\lambda_j\ketbra{u_j},
\qquad
\lambda_1\ge \lambda_2\ge\cdots\ge \lambda_d>0,
\end{equation}
be the spectral decomposition on the support of \(C_\phi\), and set
\begin{equation}
V_C\coloneqq \Tr(C_\phi)=\sum_{j=1}^{d}\lambda_j,
\qquad
\rho_C\coloneqq \frac{C_\phi}{V_C}.
\end{equation}
Define
\begin{equation}
E_K
\coloneqq
\frac{\sum_{j=1}^{K}\lambda_j}{V_C},
\qquad
E_0\coloneqq 0.
\end{equation}
For a target \(\vartheta\in(0,1)\), let
\begin{equation}
K_\vartheta
\coloneqq
\min\{K:E_K\ge \vartheta\},
\end{equation}
and define
\begin{equation}
a_\vartheta
\coloneqq
\frac{\vartheta-E_{K_\vartheta-1}}
{E_{K_\vartheta}-E_{K_\vartheta-1}}
=
\frac{\vartheta V_C-\sum_{j=1}^{K_\vartheta-1}\lambda_j}
{\lambda_{K_\vartheta}}
\in(0,1].
\end{equation}
Let \(M_{T,\vartheta}=M_{T,\mu_T}\), where \(\mu_T\) is determined by
\begin{equation}
\Tr(\rho_C M_{T,\mu_T})=\vartheta.
\end{equation}
Assume that the critical eigenvalue \(\lambda_{K_\vartheta}\) is simple. Then,
as \(T\to0\),
\begin{equation}
\label{eq:variance-profile-zero-temp-limit}
M_{T,\vartheta}
\longrightarrow
M_\vartheta
\coloneqq
P_{K_\vartheta-1}
+
a_\vartheta\ketbra{u_{K_\vartheta}},
\end{equation}
where
\begin{equation}
P_{K_\vartheta-1}
\coloneqq
\sum_{j=1}^{K_\vartheta-1}\ketbra{u_j}.
\end{equation}
The limiting filter satisfies
\begin{equation}
\Tr(\rho_C M_\vartheta)=\vartheta
\end{equation}
and solves the minimum-trace retained-variance problem
\begin{equation}
\label{eq:min-trace-retained-variance}
\min_{0\preceq M\preceq \I}\Tr(M)
\qquad
\text{subject to}
\qquad
\Tr(\rho_C M)\ge \vartheta.
\end{equation}
If the output is restricted to orthogonal projectors, the smallest feasible
principal subspace is \(P_{K_\vartheta}\).
\end{proposition}

\begin{proof}
Write
\begin{equation}
M_{T,\mu}
=
\sum_{j=1}^{d}
m_j(T,\mu)\ketbra{u_j},
\qquad
m_j(T,\mu)
=
\frac{1}{1+e^{(\mu-\lambda_j)/T}}.
\end{equation}
The covariance-probe condition is
\begin{equation}
\sum_{j=1}^{d}
\frac{\lambda_j}{V_C}
m_j(T,\mu_T)
=
\vartheta.
\end{equation}
For every \(T>0\), the left-hand side is continuous and strictly decreasing in
\(\mu_T\), so \(\mu_T\) is unique.

As \(T\to0\), the logistic occupations approach step occupations. Since
\(\vartheta\in(E_{K_\vartheta-1},E_{K_\vartheta}]\), the modes
\(j<K_\vartheta\) become fully occupied and the modes \(j>K_\vartheta\) become
unoccupied. The remaining occupation is fixed by the constraint:
\begin{equation}
m_{K_\vartheta}(T,\mu_T)
\longrightarrow
a_\vartheta.
\end{equation}
This proves Eq.~\eqref{eq:variance-profile-zero-temp-limit}.

It remains to identify the limiting variational problem. By pinching any
feasible \(M\) in the eigenbasis of \(C_\phi\), one obtains a diagonal feasible
operator with the same values of \(\Tr(M)\) and \(\Tr(\rho_C M)\). Thus the
problem reduces to
\begin{equation}
\min_{0\le m_j\le1}\sum_{j=1}^{d}m_j
\qquad
\text{subject to}
\qquad
\sum_{j=1}^{d}\frac{\lambda_j}{V_C}m_j\ge \vartheta.
\end{equation}
Because the weights \(\lambda_j/V_C\) are ordered nonincreasingly, the optimal
solution fills the largest weights first. Hence
\begin{equation}
m_j=1\quad(j<K_\vartheta),
\qquad
m_{K_\vartheta}=a_\vartheta,
\qquad
m_j=0\quad(j>K_\vartheta),
\end{equation}
which gives exactly \(M_\vartheta\). If \(M\) is restricted to projectors, the
fractional occupation is not allowed, and the smallest feasible projector is
\(P_{K_\vartheta}\).
\end{proof}

The covariance-probe scan is a variance-indexed view of the same trace-constrained path \(M_{T,\mu}\). It is distinct from the finite-temperature fixed-variance variational problem in Appendix~\ref{app:fixed-variance-soft-pca}, whose optimizer has the form
\begin{equation}
M_\lambda^{\rm var}
=
\left(\I+e^{(\I-\lambda C_\phi)/T'}\right)^{-1}.
\end{equation}
For fixed \(T'\), changing \(\lambda\) rescales the covariance operator and generally requires a \(\lambda\)-dependent evolution, or equivalently a change in effective temperature. We therefore use the trace-constrained path as the main algorithmic setting and keep the fixed-variance formulation as an appendix variant. Using bisection search, this calibration for a given target variance level $\gamma_{\rm var}$ costs $S=O(\frac{\log(1/\eta_{\rm var})}{\eta_{\rm var}^2})$ samples at additive accuracy $\eta_{\rm var}$ according to Proposition~\ref{prop:fixed-variance-complexity}.

These two approaches agree in the
low-temperature explained-variance limit: both recover the minimum-trace
fractional filter in Eq.~\eqref{eq:variance-profile-zero-temp-limit}, and when
restricted to projectors, the same smallest principal subspace
\(P_{K_\vartheta}\).

\subsection{Quantum soft PCA inference}
After calibration, inference on new input states uses the same fixed thermal
measurement circuit
\begin{equation}
U_C=e^{-i\hat p\otimes C_\phi/T_2}.
\end{equation}
The calibration stage may provide thresholds in two useful forms. A
rank-indexed threshold \(\beta_K^\star\) selects the soft filter
\begin{equation}
M_{T,K}^\star
=
M_{T,\beta_K^\star T_2},
\qquad
\Tr(M_{T,K}^\star)=K,
\end{equation}
whereas a variance-indexed threshold \(\beta^\star_\vartheta\) selects
\begin{equation}
M^\star_{T,\vartheta}
=
M_{T,\beta^\star_\vartheta T_2},
\qquad
\Tr(\rho_C M^\star_{T,\vartheta})=\vartheta,
\end{equation}
where \(\vartheta\in(0,1)\) is the normalized retained-variance level on the training covariance.

It is useful to write both cases with a generic calibrated label \(\ell\):
\begin{equation}
\beta_\ell
=
\begin{cases}
\beta_K^\star, & \ell=K,\\
\beta_\vartheta^\star, & \ell=\vartheta,
\end{cases}
\qquad
M_{T,\ell}
\coloneqq
M_{T,\beta_\ell T_2}.
\end{equation}
Thus \(\ell\) specifies how the threshold was calibrated, while the deployment
measurement itself is always the same tail event \(q>\beta_\ell\).

For a new input \(\mathbf x\), prepare the centered feature vector
\begin{equation}
\ket{z_{\mathbf x}}
\coloneqq
\ket{\phi(\mathbf x)}-\ket m,
\qquad
\nu_{\mathbf x}
\coloneqq
\braket{z_{\mathbf x}}{z_{\mathbf x}},
\end{equation}
using the LCU centered-input preparation described in
Section~\ref{sec:qcovariance}. Conditioned on successful preparation, the data
register contains
\begin{equation}
\ket{\widetilde z_{\mathbf x}}
=
\frac{\ket{z_{\mathbf x}}}{\sqrt{\nu_{\mathbf x}}},
\qquad
\nu_{\mathbf x}>0.
\end{equation}
If \(\nu_{\mathbf x}=0\), then the input coincides with the feature-space mean,
and all unnormalized centered spectral scores are zero.

Using the normalized centered input as the probe,
\begin{equation}
\rho_{\rm in}=\ketbra{\widetilde z_{\mathbf x}},
\end{equation}
the calibrated tail event gives
\begin{equation}
\label{eq:generic-tail-score}
\Pr(q>\beta_\ell\mid \ketbra{\widetilde z_{\mathbf x}})
=
\Tr\!\left(M_{T,\ell}\ketbra{\widetilde z_{\mathbf x}}\right)
=
\frac{s_{T,\ell}(\mathbf x)}{\nu_{\mathbf x}},
\end{equation}
where
\begin{equation}
\label{eq:generic-unnormalized-score}
s_{T,\ell}(\mathbf x)
\coloneqq
\bra{z_{\mathbf x}}M_{T,\ell}\ket{z_{\mathbf x}}.
\end{equation}
Accordingly,
\begin{equation}
\label{eq:generic-normalized-score}
\bar s_{T,\ell}(\mathbf x)
\coloneqq
\frac{s_{T,\ell}(\mathbf x)}{\nu_{\mathbf x}}
\end{equation}
is the normalized soft principal score of the test input with respect to the
calibrated filter \(M_{T,\ell}\). When \(\ell=K\), this is the soft rank-\(K\)
principal-subspace score. When \(\ell=\vartheta\), this is the score obtained
from the filter selected to retain training covariance fraction \(\vartheta\).
The residual scores are
\begin{equation}
1-\bar s_{T,\ell}(\mathbf x)
\end{equation}
in the normalized case and
\begin{equation}
\nu_{\mathbf x}-s_{T,\ell}(\mathbf x)
\end{equation}
in the unnormalized case.

The rank-indexed calibration also gives a soft spectral energy profile. For the
rank ladder, define
\begin{equation}
M_{T,0}^\star\coloneqq 0,
\qquad
M_{T,d}^\star\coloneqq \I,
\qquad
D_{T,k}
\coloneqq
M_{T,k}^\star-M_{T,k-1}^\star.
\end{equation}
The corresponding windows are
\begin{equation}
W_1=(\beta_1^\star,\infty),
\qquad
W_k=(\beta_k^\star,\beta_{k-1}^\star],
\quad k=2,\dots,d-1,
\qquad
W_d=(-\infty,\beta_{d-1}^\star].
\end{equation}
For the normalized centered input, the measured window probability is
\begin{equation}
\Pr(q\in W_k\mid \ketbra{\widetilde z_{\mathbf x}})
=
\Tr\!\left(D_{T,k}\ketbra{\widetilde z_{\mathbf x}}\right)
=
\frac{\bra{z_{\mathbf x}}D_{T,k}\ket{z_{\mathbf x}}}{\nu_{\mathbf x}}.
\end{equation}
Thus the normalized soft spectral energy profile is
\begin{equation}
\pi_{T,k}(\mathbf x)
\coloneqq
\Pr(q\in W_k\mid \ketbra{\widetilde z_{\mathbf x}})
=
\frac{e_{T,k}(\mathbf x)}{\nu_{\mathbf x}},
\qquad
e_{T,k}(\mathbf x)
\coloneqq
\bra{z_{\mathbf x}}D_{T,k}\ket{z_{\mathbf x}}.
\end{equation}
Since
\begin{equation}
\sum_{k=1}^d D_{T,k}=\I,
\end{equation}
the profile satisfies
\begin{equation}
\pi_{T,k}(\mathbf x)\ge0,
\qquad
\sum_{k=1}^d \pi_{T,k}(\mathbf x)=1.
\end{equation}
For a rank budget \(K\), the cumulative score is recovered as
\begin{equation}
\bar s_{T,K}(\mathbf x)
=
\sum_{k=1}^K\pi_{T,k}(\mathbf x),
\qquad
s_{T,K}(\mathbf x)
=
\sum_{k=1}^K e_{T,k}(\mathbf x).
\end{equation}
For a variance-indexed threshold \(\ell=\vartheta\), the tail score
\(\bar s_{T,\vartheta}(\mathbf x)\) in
Eq.~\eqref{eq:generic-normalized-score} is defined directly by the calibrated
threshold \(\beta_\vartheta\); it need not coincide with a sum over an
integer number of rank-indexed bins.

In practice, one may use the unpostselected centered LCU state
\(\ket{\Xi_{\mathbf x}}\) directly and jointly measure the LCU success
registers and the position \(q\). Since the centered success branch has
probability \(\nu_{\mathbf x}/4\), the joint probabilities yield the
unnormalized quantities:
\begin{equation}
\Pr(\mathrm{success},\,q\in W_k\mid \ketbra{\Xi_{\mathbf x}})
=
\frac14\bra{z_{\mathbf x}}D_{T,k}\ket{z_{\mathbf x}}
=
\frac14 e_{T,k}(\mathbf x),
\end{equation}
and for any calibrated tail label \(\ell\),
\begin{equation}
\label{eq:generic-joint-tail-test}
\Pr(\mathrm{success},\,q>\beta_\ell\mid \ketbra{\Xi_{\mathbf x}})
=
\frac14\bra{z_{\mathbf x}}M_{T,\ell}\ket{z_{\mathbf x}}
=
\frac14 s_{T,\ell}(\mathbf x).
\end{equation}
The normalized score can therefore be obtained as the conditional probability
\begin{equation}
\bar s_{T,\ell}(\mathbf x)
=
\frac{
\Pr(\mathrm{success},\,q>\beta_\ell\mid \ketbra{\Xi_{\mathbf x}})
}{
\Pr(\mathrm{success}\mid \ketbra{\Xi_{\mathbf x}})
}.
\end{equation}
Thus deployment only requires centered-input preparation, the fixed evolution
\(U_C\), and classical binning or thresholding of the measured position outcome.
No rank-dependent circuit update or further optimization is required at
deployment time.

To interpret the measured profile, write
\begin{equation}
C_\phi
=
\sum_{i=1}^d \lambda_i\ketbra{u_i},
\qquad
\ket{\widetilde z_{\mathbf x}}
=
\sum_{i=1}^d c_i(\mathbf x)\ket{u_i}.
\end{equation}
In the nondegenerate low-temperature regime, the rank-indexed soft component
effects approach the rank-one principal projectors. More precisely, using Eq.~\eqref{eq:Dk-rank-one-approx}
\begin{equation}
\|D_{T,k}-\ketbra{u_k}\|_{\rm op}
\le
\varepsilon_{T,k}+\varepsilon_{T,k-1},
\end{equation}
we obtain
\begin{equation}
\left|
\pi_{T,k}(\mathbf x)
-
|c_k(\mathbf x)|^2
\right|
\le
\varepsilon_{T,k}+\varepsilon_{T,k-1}.
\end{equation}
Equivalently,
\begin{equation}
\left|
e_{T,k}(\mathbf x)
-
\nu_{\mathbf x}|c_k(\mathbf x)|^2
\right|
\le
\nu_{\mathbf x}
\left(
\varepsilon_{T,k}+\varepsilon_{T,k-1}
\right).
\end{equation}
Hence the bin probabilities recover the centered PCA spectral energy
decomposition in the low-temperature limit.

For a variance-indexed threshold, the low-temperature limit has the
explained-variance interpretation of
Proposition~\ref{prop:low-temp-explained-variance}. In particular, if
\begin{equation}
M_{T,\vartheta}
\longrightarrow
M_\vartheta
=
P_{K_\vartheta-1}
+
a_\vartheta\ketbra{u_{K_\vartheta}},
\end{equation}
then
\begin{equation}
\bar s_{T,\vartheta}(\mathbf x)
\longrightarrow
\sum_{i=1}^{K_\vartheta-1}|c_i(\mathbf x)|^2
+
a_\vartheta |c_{K_\vartheta}(\mathbf x)|^2.
\end{equation}
If one restricts to projector-valued PCA subspaces, the corresponding smallest
principal subspace is \(P_{K_\vartheta}\), and the limiting score is
\begin{equation}
\sum_{i=1}^{K_\vartheta}|c_i(\mathbf x)|^2.
\end{equation}

Finally, for completeness, we record the data-register state produced by the thermal measurement circuit when the input-as-probe mode is used. 
The position-outcome Kraus density is given as
\begin{equation}
K_q
=
\sum_{i=1}^d
\sqrt{g_{T_1}\!\left(q-\frac{\lambda_i}{T_2}\right)}
\ketbra{u_i}.
\end{equation}
For a position window \(W\), the postselected state is
\begin{equation}
\rho_{\rm out\mid W}
=
\frac{\mathcal I_W(\rho_{\rm in})}
{\Tr(E(W)\rho_{\rm in})},
\end{equation}
where
\begin{equation}
\mathcal I_W(\rho)
=
\int_W K_q\rho K_q^\dagger\,dq,
\qquad
E(W)
=
\int_W K_q^\dagger K_q\,dq.
\end{equation}
For a calibrated tail window
\begin{equation}
W_{> \ell}\coloneqq(\beta_\ell,\infty),
\end{equation}
and \(\rho_{\rm in}=\ketbra{\widetilde z_{\mathbf x}}\), this becomes
\begin{equation}
\rho_{\rm out\mid \ell}
=
\frac{
\int_{\beta_\ell}^{\infty}
K_q\ketbra{\widetilde z_{\mathbf x}}K_q^\dagger\,dq
}{
\Tr(M_{T,\ell}\ketbra{\widetilde z_{\mathbf x}})
}
=
\frac{
\sum_{i,j=1}^d
\Gamma_{ij}^{(\ell)}
c_i(\mathbf x)\overline{c_j(\mathbf x)}
\ket{u_i}\bra{u_j}
}{
\sum_{i=1}^d
m_i^{(\ell)}|c_i(\mathbf x)|^2
},
\end{equation}
where
\begin{equation}
\Gamma_{ij}^{(\ell)}
\coloneqq
\int_{\beta_\ell}^{\infty}
\sqrt{
g_{T_1}\!\left(q-\frac{\lambda_i}{T_2}\right)
g_{T_1}\!\left(q-\frac{\lambda_j}{T_2}\right)
}\,dq,
\end{equation}
and
\begin{equation}
m_i^{(\ell)}
=
\Gamma_{ii}^{(\ell)}
=
\frac{1}{1+\exp\!\left(\frac{\beta_\ell T_2-\lambda_i}{T}\right)}.
\end{equation}
Therefore the diagonal weights are updated as
\begin{equation}
\langle u_i|\rho_{\rm out\mid \ell}|u_i\rangle
=
\frac{
m_i^{(\ell)}|c_i(\mathbf x)|^2
}{
\sum_{j=1}^d m_j^{(\ell)}|c_j(\mathbf x)|^2
}.
\end{equation}
The off-diagonal entries are weighted by window-overlap factors depending on
pairs of eigenvalues. Therefore, the accepted branch is a
measurement-induced soft spectral filter of the input. It should not be viewed
as the coherent action of \(M_{T,\ell}\) on the input vector, nor identified
with a coherent low-dimensional PCA encoding.

\begin{remark}
The protocol recovers soft principal-subspace scores and squared spectral
energies. It does not recover the relative signs or phases of the principal
coefficients \(c_i(\mathbf x)\). Consequently, it does not directly yield the
coherent vector-valued PCA reconstruction
\begin{equation}
\sum_{i=1}^K c_i(\mathbf x)\ket{u_i}.
\end{equation}
Obtaining such a reconstruction would require additional phase-sensitive
subroutines.
\end{remark}

\section{Comparison with Existing PCA Methods}
\label{sec:comparison}

A meaningful comparison between PCA algorithms must first specify the
downstream task, because different algorithms produce structurally different
outputs. 
Existing classical and
quantum PCA protocols generally target explicit spectral reconstruction: they return principal eigenvectors, approximated eigenstates, or coherent spectral projectors. 
Our protocol instead produces the entropy-regularized, measurement-based analog of such low-rank objects: a soft
filter family \(\{M_{T,k}^\star\}_{k=1}^{d-1}\), the associated soft spectral-component POVM \(\{D_{T,k}\}_{k=1}^d\), or the retained-variance profiles \(\{M_{T,\vartheta}^\star\}_{\vartheta}\). These objects are accessed not as classical matrices to multiply classical
vectors, but as quantum-measurement effects acting on quantum input states. The
natural deliverables are scalar scores
\(\Tr(M_{T,k}^\star\rho_{\mathrm{in}})\), retained-variance test scores
\(\Tr(M_{T,\vartheta}^\star\rho_{\mathrm{in}})\), and postselected quantum states. Direct cost
comparison with other PCA algorithms is meaningful on tasks that all algorithms
can solve in their native outputs, such as principal-subspace scoring,
cumulative-variance estimation, and postselection. It is qualitative on tasks
targeting outputs, such as eigenvector lists or coherent reconstructions, for
which our protocol would require an additional tomographic step. Table~\ref{tab:pca-comparison} summarizes the resulting comparison.

\begin{table}
\centering
\small
\begin{tabular}{m{2.7cm}m{3.5cm}m{4.5cm}m{3.5cm}}
\hline
Algorithm & Output & Cost & Direct comparison with this work? \\
\hline
Classical SVD / Krylov~\cite{GolubVanLoan2013,Saad2011}
  & Eigenvectors $\{u_j\}$
  & $O(Nd\min(N,d))$ to $\widetilde{O}(NdK)$
  & No: different output type \\
Randomized SVD~\cite{HalkoMartinssonTropp2011}
  & Eigenvectors $\{u_j\}$
  & $\widetilde{O}(NdK)$
  & No: different output type \\
Streaming / online PCA~\cite{Shamir2016,Jain2016}
  & Eigenvectors $\{u_j\}$
  & $O(dK)$ per update
  & No: different output type \\
Lloyd qPCA $+$ QPE~\cite{Lloyd2014}
  & Quantum eigenstates $\ket{u_j}$, eigenvalue estimates
  & $\widetilde{O}(R^2\kappa^2/\epsilon^3)$ copies
  & Yes: on scoring tasks \\
QSVT-based PCA~\cite{Gilyen2019}
  & Coherent rank-$K$ projector on quantum state
  & $\widetilde{O}(K\kappa/\Delta_K)$ queries
  & Yes: on scoring / projection tasks \\
Variational qPCA~\cite{LaRose2019}
  & Quantum eigenstates via ansatz
  & No general guarantee
  & Partial: depends on convergence \\
Tang dequantization~\cite{Tang2019,Tang2021QPCA}
  & Classical analog of qPCA outputs
  & $\widetilde{O}(\mathrm{poly}(K,1/\epsilon,\kappa))$
  & No: requires classical sample-and-query access \\
\textbf{This work}
  & Soft score $\Tr(M_{T,k}^\star\rho_{\mathrm{in}})$, $\Tr(M_{T,\vartheta}^\star\rho_{\mathrm{in}})$ or postselected state
  & $S=O(1/\eta^{2})$ to $O(d^2/\eta^{2})$ samples $+$ $\widetilde{O}(1)$ per query
  & \multicolumn{1}{c}{---} \\
\hline
\end{tabular}
\caption{Cost comparison across PCA algorithms. The $\widetilde{O}(\cdot)$ notation hides poly-logarithmic factors. The ``Direct comparison'' column indicates whether a task exists for which both algorithms can produce the same operational output, allowing head-to-head cost comparison. For algorithms targeting eigenvector lists, comparison with this work requires either a tomographic step (cancelling the sample-complexity advantage) or a restriction to scoring tasks.}
\label{tab:pca-comparison}
\end{table}

\paragraph{Classical PCA.}
Given a centered data matrix of $N$ samples in $d$ dimensions, direct singular-value decomposition has cost
\(O(Nd\min(N,d))\)~\cite{GolubVanLoan2013}. Truncated decompositions to the
top-\(K\) subspace are usually computed by Krylov methods, such as Lanczos or
subspace iteration, at cost
\(O(NdK\log(1/\epsilon)/\sqrt{\Delta_K})\), where
\(\Delta_K=\lambda_K-\lambda_{K+1}\) is the spectral gap~\cite{Saad2011}.
Randomized SVD reduces this to \(O(Nd\log K+(N+d)K^2)\)~\cite{HalkoMartinssonTropp2011}.
Streaming and stochastic algorithms, such as Oja's rule and
matrix-stochastic-gradient methods, lower the per-update cost to \(O(dK)\) at
the price of multiple passes over the data~\cite{Shamir2016,Jain2016}. In all
cases, the output is an explicit set of \(K\) classical \(d\)-dimensional
vectors or an explicit subspace estimate. Subsequent scoring of a test point
\(\mathbf y\in\R^d\) then incurs an additional \(O(dK)\) matrix-vector cost.
The same issue appears for retained-variance or model-order diagnostics. In
classical PCA, the explained-variance curve is normally computed from the
principal spectrum:
\begin{equation}
K\longmapsto
\frac{\sum_{j=1}^K\lambda_j}{\sum_{j=1}^d\lambda_j}.
\end{equation}
Thus, even if the downstream task only asks for the smallest subspace retaining
a prescribed variance fraction, the standard route first computes or
approximates the leading spectral data. 

\paragraph{Quantum PCA.}
Lloyd, Mohseni, and Rebentrost~\cite{Lloyd2014} introduced quantum PCA by
simulating \(e^{-i\rho t}\) from \(O(t^2/\epsilon)\) copies of the density
operator \(\rho=C/\Tr(C)\) and applying quantum phase estimation to read out
eigenphases on principal eigenvectors, at total query cost polynomial in the
effective rank \(R\) of \(\rho\), the condition number \(\kappa\), and the
inverse accuracy \(1/\epsilon\). This is typically reported as
\(\widetilde{O}(R^2\kappa^2/\epsilon^3)\) in regimes where these parameters are
well controlled. Block-encoding and QSVT methods~\cite{Gilyen2019} project onto
top-\(K\) subspaces using \(\widetilde{O}(K\kappa/\Delta_K)\) queries to a
covariance block encoding, with polylogarithmic dependence on \(d\) when
QRAM-like access is granted. Variational quantum
PCA~\cite{LaRose2019} diagonalizes a parameterized ansatz against \(C\) via
gradient-based optimization, and currently lacks general complexity guarantees.
Tang's dequantization results~\cite{Tang2019,Tang2021QPCA} establish that
whenever such oracle access is granted to a classical algorithm with
sample-and-query data structures, top-\(K\) PCA can be performed in classical
time \(\widetilde{O}(\mathrm{poly}(K,1/\epsilon,\kappa))\), independent of
\(d\). The figures quoted here should be read as indicative scalings rather than tight bounds, since the precise complexity depends on the oracle model, normalization, spectral gap, and success-probability assumptions.
All of these protocols target spectral information: eigenvectors, eigenstates,
eigenvalues, or coherent spectral projectors. They can therefore be used to
derive principal scores and retained-variance diagnostics, but generally only
after resolving a spectral object.

\paragraph{Where the proposed framework sits.}
Our protocol replaces spectral reconstruction with a single calibrated measurement. After fixing the evolution $U_C=e^{-i\hat p\otimes C_\phi/T_2}$, the optimal thresholds are recovered via classical quantile lookup, yielding a total cost breakdown of:
\begin{enumerate}[label=(\roman*),leftmargin=1.5em]

\item one-time calibration: \(O(d^2/\eta^2)\) samples for a uniform
all-rank unnormalized trace guarantee; \(O(K(d-K)/\eta^2)\) for one fixed
rank; \(O(\eta^{-2})\) for normalized fractional trace calibration
(Theorem~\ref{thm:quantile-soft-pca}, Proposition~\ref{prop:fixed-rank-sampling}) and for the full retained-variance profile (If raw centered-LCU attempts rather than successful covariance-probe samples are counted, the retained-variance cost is multiplied by \(4/\Tr(C_\phi)\), Proposition~\ref{prop:cov-probe-profile-complexity});
\item per-sample quantum cost:
\(\widetilde{O}(1/(T_2^2\varepsilon_{\mathrm{sim}}))\) selector-swap micro-steps
for the Hamiltonian simulation (Corollary~\ref{cor:gate-complexity});
\item per-input deployment: one thermal measurement on the test state, with no
further training or recompilation.
\end{enumerate}

This profile is best matched to the following regimes:
\begin{itemize}[leftmargin=1.5em]
\item \emph{Score-based Tasks on Quantum Inputs.} For downstream evaluations like anomaly scoring or subspace projection, our protocol produces the required scalar directly. This bypasses the structural overhead of classical PCA (explicit vectors) and eigenspectral qPCA (phase estimation), rendering the deployment cost entirely independent of the target rank $K$.
\item \emph{Dimension-Independent Normalized Profiling.} With the maximally mixed or covariance probe, the protocol natively learns normalized soft trace fractions or retained-variance profiles as direct Born probabilities. This achieves additive accuracy $\eta$ with $O(\eta^{-2})$ samples, sidestepping the standard requirement to first diagonalize the spectrum to compute explained variance.
\item \emph{Training-Free Multi-Target Deployment.} Unlike variational qPCA which necessitates a separate optimization loop, or QSVT which requires recompiling singular-value polynomials for different $K$, our single empirical distribution calibrates the entire rank-indexed or variance-indexed filter ladder simultaneously. Subsequent queries simply adjust a classical threshold.
\item \emph{Operator-Level Quantum-Kernel PCA.} In standard quantum-kernel workflows, forming the classical Gram matrix requires $\Theta(N^2)$ pairwise quantum evaluations. Our framework avoids this bottleneck. Because the nonzero spectrum of $C_\phi$ coincides with that of the centered kernel matrix $K_c/N$ (Remark~\ref{rem:kernel}), the protocol accesses the kernel-PCA spectrum strictly at the operator level without computing or storing individual matrix entries.
\item \emph{Native Quantum Centering.} Classical preprocessing and standard kernel PCA require explicit mean subtraction. When data are quantum-native or nonlinearly embedded, classical centering may be ill-defined or inaccessible. Our protocol integrates centering natively into the state preparation step (via the subnormalized mean branch $\ket{\chi_m}$), maintaining the entire procedure within the quantum access model without demanding state tomography.
\end{itemize}

The framework is, conversely, not advantageous when an explicit classical list
of principal eigenvectors is required, since extracting such vectors from the
postselected soft component states \(D_{T,k}\) requires tomography that cancels
the quantum-native sample-complexity benefit. Similarly, the covariance-probe
scan gives the retained-variance profile directly, but it does not by itself
give a classical eigenspectrum. Comparisons with dequantized classical PCA~\cite{Tang2019,Tang2021QPCA} are
also delicate: dequantization assumes structured classical sample-and-query
access to the data matrix, while our protocol operates on quantum input states
that may be produced by a quantum process, where no such classical handle exists.

% Table~\ref{tab:pca-comparison} summarizes the picture across the algorithms
% considered. The rightmost column reports whether a downstream task exists for
% which both algorithms can produce the same operational output, allowing
% head-to-head cost benchmarking. For algorithms targeting eigenvector lists,
% comparison with the present work requires either a tomographic conversion step,
% which cancels the sample-complexity advantage, or a restriction to score-based
% and retained-variance tasks.

\section{Discussion and Outlook}
\label{sec:4}

The main observation in this work is that, for score-based tasks, PCA can be treated as a measurement problem rather than an eigenvector-recovery problem. This reformulation is naturally compatible with quantum computation, where \(0\preceq M\preceq \I\) is a measurement effect and \(\Tr(M\rho)\) is a Born probability.
Adding Fermi--Dirac entropy replaces the hard projector with a unique soft filter whose occupations vary smoothly between zero and one. 
Finite temperature smooths the spectral cutoff, while the zero-temperature limit recovers hard PCA under a nonzero eigengap. 

This distinction is operationally important for quantum data. 
The protocol does not need to output a classical list of principal eigenvectors. 
Instead, rank-dependent soft PCA filters are selected by changing only a classical position threshold after one fixed thermal measurement circuit. 
These thresholds are calibrated by quantiles of a single position distribution, rather than by a rank-dependent hybrid optimization loop. 
An additional operational advantage is that centering is not delegated to a classical preprocessing stage. The empirical mean branch state is prepared coherently. It centers the training covariance and each test input during inference. This is especially relevant for quantum-native data and nonlinear quantum feature maps, where the feature-space mean may not be available as a classical vector.

The cost of the resource depends on the task. 
Uniform calibration of all ranks requires estimating a full cumulative distribution and carries the worst case \(d^2\)-type sampling factor. For one fixed target rank \(K\), the sampling cost is smaller and scales with the binomial variance of the corresponding tail probability. 
For normalized fraction-rank or retained-variance targets, the statistical cost of sampling can be independent of \(d\). 
Thus, the protocol is most useful when the downstream task consumes soft scores, residual scores, or postselected quantum states, rather than a full classical eigendecomposition. 

Likewise, obtaining a classical description of the soft component operators \(D_{T,k}\) requires tomography of the window conditioned output states, which can remove the advantage when a full classical model is required. The method is therefore best viewed as a quantum scoring and filtering primitive, not as a general purpose replacement for classical PCA diagonalization. 
% Sec.~\ref{sec:comparison} sets these costs against existing classical and quantum PCA algorithms in detail.

Several directions remain open. 
The Hamiltonian simulation \(U_C=e^{-i\hat p\otimes C/T_2}\) used here could be replaced by block encoding or quantum singular value transformation implementations when stronger access to the covariance operator is available. The temperature could also be chosen adaptively according to the observed spectral gaps or downstream anomaly detection performance. More broadly, the soft component POVM \(\{D_{T,k}\}\) suggests a measurement-based analogue of spectral feature extraction, where principal components are not accessed as explicitly reconstructed vectors. 

\section*{Acknowledgements}
MM and MMW
acknowledge support from the Cornell School of Electrical and Computer Engineering. NL acknowledges funding from the
Science and Technology Commission of Shanghai Municipality (STCSM)
grant no.~24LZ1401200 (21JC1402900), NSFC grants no.~12471411 and
no.~12341104, the Shanghai Jiao Tong University 2030 Initiative, the Shanghai Pilot Program for Basic Research, 
and the Fundamental Research Funds for the Central Universities.

\bibliographystyle{unsrturl}
\bibliography{soft_pca_top_style_refs}

\newpage
\appendix
\section{Reconstruction error of fractional Fantope optimizers}
\label{app:fantope-shrinkage}
Viewing the trace constraint as a convex relaxation over the Fantope provides a natural framework for analyzing the data reconstruction error. For a centered random data vector $\mathbf y$ with covariance $C$, applying the measurement operator $M \in \mathcal{M}_d$ as a linear filter yields the reconstruction $\hat{\mathbf y} = M\mathbf y$. The expected mean-squared reconstruction error is given by
\begin{align}
\mathbb E\|\mathbf y-M\mathbf y\|^2
&=
\mathbb E\!\left[
\mathbf y^\dagger(\I-M)^\dagger(\I-M)\mathbf y
\right]\\
&=
\Tr\!\left((\I-M)^\dagger(\I-M)C\right)\\
&=
\Tr(C)-2\operatorname{Re}\Tr(CM)+\Tr(CM^\dagger M).
\end{align}
For the measurement effects considered in the main text, \(M=M^\dagger\), and hence
\begin{equation}
\label{eq:reconstruction_error}
  \mathcal E(M)
=
\Tr(C)-2\Tr(CM)+\Tr(CM^2).  
\end{equation}

For a rank-\(k\) projector \(P_k\), the strict algebraic property $P_k^2 = P_k$ holds. Substituting this into Eq.~\eqref{eq:reconstruction_error}, the classical PCA reconstruction error is
\begin{equation}
\mathcal E(P_k)
=
\Tr(C)-\Tr(CP_k).
\end{equation}

When the boundary eigengap closes ($\lambda_k = \lambda_{k+1}$), the hard trace-constrained problem may admit fractional Fantope optimizers. 
For such \(M^\star\), one has $M^{\star2}\preceq M^\star$,
with strict inequality on the fractional eigenspaces. 
If \(C\) has positive weight on those fractional eigenspaces, then $\Tr(CM^{\star2})<\Tr(CM^\star)$. Because both the fractional $M^\star$ and the projector $P_k$ attain the same maximal linear variance $\Tr(CM^\star) = \Tr(CP_k) = \sum_{j=1}^k \lambda_j$, the reconstruction error of the fractional solution satisfies
\begin{equation}
\mathcal{E}(M^\star) = \Tr(C) - 2\Tr(CM^\star) + \Tr(C M^{\star 2}) < \Tr(C) - \Tr(C M^\star) = \mathcal{E}(P_k).
\end{equation}

Under the relaxed linear filtering objective, fractional Fantope optimizers may reduce Euclidean shrinkage error, but they no longer represent rank-k reconstructions. This arises because the trace constraint $\Tr(M)=k$ relaxes the strict rank constraint $\rank(M)=k$. The fractional solution distributes its spectrum across more than $k$ dimensions ($\rank(M^\star) > k$), providing a `soft' filtering that reduces the Euclidean error.

\section{Classical Optimization Algorithms for Soft PCA}
\label{app:classicalopt}
Recall that we have the strictly convex one-dimensional dual function Eq.~\eqref{eq:scalar-dual-objective}
\begin{equation}
\Phi_T(\mu)
=
\mu k + T\Tr\log\!\left(\I+e^{(C-\mu\I)/T}\right).
\end{equation}
corresponding to the soft PCA in Proposition~\ref{thm:soft-pca-scalar-specialization} with its derivative and second derivative Eq.~\eqref{eq:scalar-dual-derivative}
\begin{equation}
\Phi_T'(\mu)=k-\Tr(M_{T,\mu}),
\qquad
\Phi_T''(\mu)=\frac{1}{T}\Tr\!\bigl[M_{T,\mu}(\I-M_{T,\mu})\bigr].
\end{equation}
Since
\begin{equation}
0\le x(1-x)\le \frac14
\qquad \text{for all } x\in[0,1],
\end{equation}
we obtain the global smoothness bound
\begin{equation}
\label{eq:Lipschitz-constant}
0<\Phi_T''(\mu)\le \frac{d}{4T}
\qquad \text{for all } \mu\in\mathbb R.
\end{equation}
Thus $\Phi_T$ is a strictly convex $L_T$-smooth function with
\begin{equation}
L_T\coloneqq \frac{d}{4T}
\end{equation}
where $L_T$ defines a Lipschitz constant for the gradient.

It is convenient to introduce an interval for the optima $\mu^\star$. Define
\begin{equation}
\underline\mu \coloneqq \lambda_d(C)+T\log\frac{d-k}{k},
\qquad
\overline\mu \coloneqq \lambda_1(C)+T\log\frac{d-k}{k},
\end{equation}
and set
\begin{equation}
I_T\coloneqq [\underline\mu,\overline\mu].
\end{equation}
Then $\mu^\star\in I_T$. Indeed, the trace function
\begin{equation}
\psi_T(\mu)\coloneqq \Tr(M_{T,\mu})
=
\sum_{j=1}^d \frac{1}{1+e^{(\mu-\lambda_j(C))/T}}
\end{equation}
is strictly decreasing, and one checks that
\begin{equation}
\psi_T(\underline\mu)\ge k,
\qquad
\psi_T(\overline\mu)\le k.
\end{equation}
As such, the standard gradient descent methods can be applied to minimize the objective in Eq.~\eqref{eq:scalar-dual-objective} with respect to $\mu$, guaranteeing first order or quadratic convergence to the global minimum, which we provide as Algorithm~\ref{alg:soft-pca-first order} and Algorithm~\ref{alg:soft-pca-second order}.
\begin{algorithm}[h]
\caption{first order gradient descent of the soft PCA}
\label{alg:soft-pca-first order}
\begin{algorithmic}[1]
\Require Desired error $\epsilon>0$, Temperature $T=\frac{\epsilon}{d\,h(k/d)}$, iteration number $N$, step size $\eta\in(0,1/L_T)$, initialization $\mu_0\in I_T$
\For{$t=0,1,\dots,N-1$}
    \State Compute the gradient
    \begin{equation*}
    g_t \coloneqq \Phi_T'(\mu_t)=k-\Tr(M_{T,\mu_t}).
    \end{equation*}
    \State Update
    \begin{equation*}
    \mu_{t+1} \coloneqq \Pi_{I_T}\!\bigl(\mu_t-\eta g_t\bigr),
    \end{equation*}
    where $\Pi_{I_T}$ denotes projection onto the interval $I_T$.
\EndFor
\State \Return $\mu_N$ and the corresponding soft filter $M_{T,\mu_N}$.
\end{algorithmic}
\end{algorithm}

\begin{theorem}[first order convergence and iteration complexity]
\label{thm:first order-soft-pca}
Let $\{\mu_t\}_{t\ge 0}$ be generated by Algorithm~\ref{alg:soft-pca-first order} with initialization $\mu_0\in I_T$ and the canonical step size $\eta=1/L_T$. Then $\mu_t\in I_T$ for all $t$, and the objective error after $N$ iterations is bounded by
\begin{equation}
\label{eq:first order-gap-spectral-range}
\Phi_T(\mu_N)-\Phi_T(\mu^\star)
\le
\frac{d\bigl(\lambda_1(C)-\lambda_d(C)\bigr)^2}{8TN}.
\end{equation}
Consequently, to achieve an $\epsilon$-accurate optimization $\Phi_T(\mu_N)-\Phi_T(\mu^\star)\le \epsilon$ under the optimal temperature choice $T\le \epsilon/\bigl(d\,h(k/d)\bigr)$ (from Corollary~\ref{cor:temperature-choice}), it is sufficient to run the algorithm for
\begin{equation}
\label{eq:N-bound-logd}
N
\ge
\frac{dk\log(d/k)\bigl(\lambda_1(C)-\lambda_d(C)\bigr)^2}{8\epsilon^2}
\end{equation}
iterations.
\end{theorem}

\begin{proof}
By the standard convergence analysis of gradient descent for $L_T$-smooth functions, the error after $N$ iterations with step size $\eta$ is bounded by
\begin{equation}
\Phi_T(\mu_N)-\Phi_T(\mu^\star)
\le
\frac{|\mu_0-\mu^\star|^2}{2\eta N}.
\end{equation}
Substituting the canonical step size $\eta=1/L_T=4T/d$, we obtain
\begin{equation}
\Phi_T(\mu_N)-\Phi_T(\mu^\star)
\le
\frac{d}{8TN}\,|\mu_0-\mu^\star|^2.
\end{equation}
Since both the initialization $\mu_0$ and the optimum $\mu^\star$ are confined within the interval $I_T$, their maximum possible distance is strictly bounded by the spectral spread of the covariance matrix:
\begin{equation}
|\mu_0-\mu^\star|
\le
\lambda_1(C)-\lambda_d(C).
\end{equation}
Plugging this uniform estimate into the error bound directly yields Eq.~\eqref{eq:first order-gap-spectral-range}.

To guarantee an error of at most $\epsilon$, we rearrange Eq.~\eqref{eq:first order-gap-spectral-range} to solve for $N$, which requires
\begin{equation}
\label{eq:N-bound-fixed-T}
N
\ge
\frac{d\bigl(\lambda_1(C)-\lambda_d(C)\bigr)^2}{8T\epsilon}.
\end{equation}
Substituting the temperature upper bound $T\le \frac{\epsilon}{d\,h(k/d)}$ into Eq.~\eqref{eq:N-bound-fixed-T}, it is sufficient to take
\begin{equation}
\label{eq:N-bound-temperature-substituted}
N
\ge
\frac{d^2 h(k/d)\bigl(\lambda_1(C)-\lambda_d(C)\bigr)^2}{8\epsilon^2}.
\end{equation}
Finally, applying Eq.~\eqref{eq:uneqt}, we arrive at the final iteration complexity bound in Eq.~\eqref{eq:N-bound-logd}.
\end{proof}

\begin{algorithm}[h]
\caption{second order Newton method of the soft PCA}
\label{alg:soft-pca-second order}
\begin{algorithmic}[1]
\Require Desired error $\epsilon>0$, Temperature $T=\frac{\epsilon}{d\,h(k/d)}$, iteration number $N$, initialization $\mu_0\in I_T$
\For{$t=0,1,\dots,N-1$}
    \State Compute
    \begin{equation*}
    g_t \coloneqq \Phi_T'(\mu_t)=k-\Tr(M_{T,\mu_t}),
    \end{equation*}
    and
    \begin{equation*}
    h_t \coloneqq \Phi_T''(\mu_t)=\frac{1}{T}\Tr\!\bigl[M_{T,\mu_t}(\I-M_{T,\mu_t})\bigr].
    \end{equation*}
    \State Perform the Newton update
    \begin{equation*}
    \mu_{t+1}\coloneqq \Pi_{I_T}\!\left(\mu_t-\frac{g_t}{h_t}\right).
    \end{equation*}
\EndFor
\State \Return $\mu_N$ and the corresponding soft filter $M_{T,\mu_N}$.
\end{algorithmic}
\end{algorithm}

\begin{theorem}[Local quadratic convergence and complexity]
\label{thm:second order-soft-pca}
Let $\{\mu_t\}_{t\ge 0}$ be generated by Algorithm~\ref{alg:soft-pca-second order}. Fix a compact interval $J\subseteq I_T$ containing $\mu^\star$, and define the local strict convexity bound
\begin{equation}
m_T \coloneqq \min_{\mu\in J}\Phi_T''(\mu)>0.
\end{equation}
Leveraging the algebraic properties of the Fermi--Dirac operator, the third derivative admits a global uniform bound $|\Phi_T'''(\mu)|\le \frac{d}{4T^2}$. Then there exists $\rho_T>0$ such that whenever $|\mu_0-\mu^\star|\le \rho_T$ and all iterates remain in $J$, the parameter error decays quadratically:
\begin{equation}
|\mu_{t+1}-\mu^\star|
\le
\frac{d}{8T^2 m_T}\,|\mu_t-\mu^\star|^2.
\end{equation}
Consequently, within this local neighborhood, reaching an $\epsilon$-accurate parameter estimation $|\mu_N-\mu^\star|\le \epsilon$ requires an iteration complexity of only
\begin{equation}
N = O\left(\log \log \frac{1}{\epsilon}\right).
\end{equation}
\end{theorem}

\begin{proof}
For the unprojected Newton step $\tilde{\mu}_{t+1} \coloneqq \mu_t - \Phi_T'(\mu_t)/\Phi_T''(\mu_t)$, a standard second order Taylor expansion for the strictly convex scalar function $\Phi_T$ yields the error bound $|\tilde{\mu}_{t+1}-\mu^\star| \le \frac{\max|\Phi_T'''|}{2m_T}|\mu_t-\mu^\star|^2$. By evaluating the exact third derivative, we have
\begin{equation*}
|\Phi_T'''(\mu)| = \frac{1}{T^2}\bigl|\Tr\bigl[M_{T,\mu}(I-M_{T,\mu})(I-2M_{T,\mu})\bigr]\bigr|.
\end{equation*}
Since the eigenvalues of $(I-2M_{T,\mu})$ are strictly bounded in $[-1,1]$, this magnitude is bounded by $\frac{1}{T}\Phi_T''(\mu)$, which is globally upper-bounded by $\frac{d}{4T^2}$. Substituting this explicit bound yields $|\tilde{\mu}_{t+1}-\mu^\star| \le \frac{d}{8T^2 m_T}|\mu_t-\mu^\star|^2$. Furthermore, since $\mu^\star \in I_T$, the Euclidean projection $\Pi_{I_T}$ is non-expansive. Thus, $|\mu_{t+1}-\mu^\star| \le |\tilde{\mu}_{t+1}-\mu^\star|$, fully preserving the quadratic decay. The $O(\log \log(1/\epsilon))$ complexity follows directly from recursively applying this bound.
\end{proof}

\section{Fixed-Variance Soft PCA}
\label{app:fixed-variance-soft-pca}
\subsection{Optimizer for the fixed-variance soft PCA}
The main text studies the trace-constrained soft PCA problem
\begin{equation}
\max_{0\preceq M\preceq \I}
\ \Tr(CM)+T\fd{M}
\qquad
\text{subject to}
\qquad
\Tr(M)=k.
\end{equation}
There is a dual operational formulation in which one prescribes the amount of covariance energy to be retained and asks for the most compact soft filter. Let \(C\succeq 0\), \(V_C\coloneqq \Tr(C)>0\), and fix a target variance level
\begin{equation}
\gamma_{\rm var}\in(0,V_C).
\end{equation}
Consider
\begin{equation}
\label{eq:fixed-variance-soft-pca}
\min_{M\in\mathcal M_d}
\ \Tr(M)-T'\fd{M}
\qquad
\text{subject to}
\qquad
\Tr(CM)=\gamma_{\rm var},
\end{equation}
where \(T'>0\) is the entropy scale for this fixed-variance problem.

\begin{proposition}[Fixed-variance soft PCA has Fermi--Dirac form]
\label{prop:fixed-variance-fd-form}
The optimizer of Eq.~\eqref{eq:fixed-variance-soft-pca} is unique and is given by
\begin{equation}
\label{eq:fixed-variance-filter}
M_{\lambda}^{\mathrm{var}}
=
\left(
\I+
e^{(\I-\lambda C)/T'}
\right)^{-1},
\end{equation}
where the scalar multiplier \(\lambda\in\mathbb R\) is uniquely determined by
\begin{equation}
\label{eq:fixed-variance-constraint}
\Tr\!\left(CM_{\lambda}^{\mathrm{var}}\right)
=
\gamma_{\rm var}.
\end{equation}
Moreover, whenever \(\lambda>0\), the same optimizer can be written as a member of the Fermi--Dirac soft PCA family:
\begin{equation}
\label{eq:fixed-variance-effective-parameters}
M_{\lambda}^{\mathrm{var}}
=
\left(
\I+
e^{(\mu_{\mathrm{eff}}\I-C)/T_{\mathrm{eff}}}
\right)^{-1},
\qquad
\mu_{\mathrm{eff}}=\frac{1}{\lambda},
\qquad
T_{\mathrm{eff}}=\frac{T'}{\lambda}.
\end{equation}
Thus the trace-constrained and fixed-variance formulations select filters from the same two-parameter Fermi--Dirac family.
\end{proposition}

\begin{proof}
Introduce a Lagrange multiplier \(\lambda\in\mathbb R\) for the variance constraint and define
\begin{equation}
\mathcal L(M,\lambda)
=
\Tr(M)-T'\fd{M}
+
\lambda\bigl(\gamma_{\rm var}-\Tr(CM)\bigr).
\end{equation}
On the open operator interval \(0\prec M\prec \I\), the Fréchet derivative is
\begin{equation}
\nabla_M\mathcal L(M,\lambda)
=
\I
-
T'\bigl(\log(\I-M)-\log M\bigr)
-
\lambda C.
\end{equation}
Setting this derivative to zero gives
\begin{equation}
T'\bigl(\log M-\log(\I-M)\bigr)
=
\lambda C-\I.
\end{equation}
Equivalently,
\begin{equation}
M(\I-M)^{-1}
=
e^{(\lambda C-\I)/T'}.
\end{equation}
Solving for \(M\) yields Eq.~\eqref{eq:fixed-variance-filter}.

To see uniqueness of the multiplier, diagonalize
\begin{equation}
C=\sum_{j=1}^d c_j\ketbra{u_j},
\qquad
c_j\ge 0.
\end{equation}
Then
\begin{equation}
M_{\lambda}^{\mathrm{var}}
=
\sum_{j=1}^d
\frac{1}{1+e^{(1-\lambda c_j)/T'}}
\ketbra{u_j}.
\end{equation}
Define
\begin{equation}
A(\lambda)
\coloneqq
\Tr(CM_{\lambda}^{\mathrm{var}})
=
\sum_{j=1}^d
\frac{c_j}{1+e^{(1-\lambda c_j)/T'}}.
\end{equation}
A direct differentiation gives
\begin{equation}
A'(\lambda)
=
\frac{1}{T'}
\sum_{j=1}^d
c_j^2
m_j(\lambda)\bigl(1-m_j(\lambda)\bigr),
\end{equation}
where
\begin{equation}
m_j(\lambda)
=
\frac{1}{1+e^{(1-\lambda c_j)/T'}}.
\end{equation}
Hence \(A'(\lambda)>0\) whenever \(C\neq 0\). Moreover,
\begin{equation}
\lim_{\lambda\to-\infty}A(\lambda)=0,
\qquad
\lim_{\lambda\to+\infty}A(\lambda)=\Tr(C)=V_C.
\end{equation}
Thus for every \(\gamma_{\rm var}\in(0,V_C)\), there is a unique \(\lambda\) satisfying Eq.~\eqref{eq:fixed-variance-constraint}.

Finally, if \(\lambda>0\), then
\begin{equation}
\frac{\I-\lambda C}{T'}
=
\frac{\mu_{\mathrm{eff}}\I-C}{T_{\mathrm{eff}}},
\qquad
\mu_{\mathrm{eff}}=\frac1\lambda,
\qquad
T_{\mathrm{eff}}=\frac{T'}{\lambda}.
\end{equation}
Substituting this identity into Eq.~\eqref{eq:fixed-variance-filter} gives Eq.~\eqref{eq:fixed-variance-effective-parameters}.
\end{proof}

In the trace-constrained formulation, one usually fixes the effective temperature \(T\) and varies only the chemical potential \(\mu\). In the fixed-variance formulation Eq.~\eqref{eq:fixed-variance-soft-pca}, fixing \(T'\) and varying the variance multiplier \(\lambda\) changes both effective parameters:
\begin{equation}
\mu_{\mathrm{eff}}=\frac1\lambda,
\qquad
T_{\mathrm{eff}}=\frac{T'}{\lambda}.
\end{equation}
Therefore, the two formulations are identical at the level of the Fermi--Dirac filter family, but not necessarily along the same fixed-temperature one-parameter slice unless the entropy scale \(T'\) is co-tuned with \(\lambda\). Since the fixed-variance entropy scale $T'$ is a rigid hyperparameter defining the regularization strength, the multiplier $\lambda$ dynamically scales the effective spectrum of the covariance operator. This spectral scaling requires active quantum simulation and cannot be bypassed via a static classical threshold shift.
In practice, the operating point for a prescribed variance level is obtained by a one-dimensional bisection search over the Lagrange multiplier \(\lambda\). We now proceed to a detailed elaboration.

\subsection{Bisection calibration using a covariance probe}
\label{app:sub_fixed_variance_calib}

The trace-constrained calibration in Section~\ref{sec:3} uses the maximally
mixed probe \(\tau_{\rm mm}=\I/d\), for which
\begin{equation}
\Pr(q>\beta\mid \tau_{\rm mm})
=
\frac{1}{d}\Tr(M_{T,\beta T_2}).
\end{equation}
Thus an \(O(1)\) additive error in the unnormalized trace requires an
\(O(1/d)\) probability accuracy. For the fixed-variance formulation this
dilution can be avoided by probing the covariance itself.

Let
\begin{equation}
V_C\coloneqq \Tr(C_\phi)=1-\alpha,
\qquad
\rho_C\coloneqq \frac{C_\phi}{V_C},
\end{equation}
assuming \(V_C>0\). Although \(V_C\) need not be known classically, the state
\(\rho_C\) is prepared automatically by the centered LCU preparation described in Eq.~\eqref{eq:rho_c}. 
The total success probability of this preparation is \(V_C/4\), up to the
constant factor from the LCU normalization.
Following the thermal measurement scheme of Proposition~\ref{prop:modified-thermal measurement}, we get the thermal measurement for the fixed-variance filter.

\begin{proposition}[Measurement mapping for the fixed-variance filter]
\label{prop:fixed-variance-measurement}
Let
\begin{equation}
V_C\coloneqq \Tr(C_\phi)=1-\alpha,
\qquad
\rho_C\coloneqq \frac{C_\phi}{V_C},
\end{equation}
and assume \(V_C>0\). Set the control parameters to
\begin{equation}
T_1=\frac{T'}{T_2},
\qquad
\beta=\frac{1}{T_2}.
\end{equation}
For any \(\lambda\in\mathbb R\), perform the modified thermal measurement with the scaled data-side evolution
\begin{equation}
U_C(\lambda)
=
e^{-i\hat p\otimes(\lambda C_\phi)/T_2}
\end{equation}
on the successfully prepared covariance probe \(\rho_C\). Then the conditional acceptance probability satisfies
\begin{equation}
\label{eq:fixed-variance-conditional-prob}
p_{\rm cond}(\lambda)
\coloneqq
\Pr(q>\beta\mid \rho_C,\lambda)
=
\Tr(\rho_C M_\lambda^{\rm var})
=
\frac{\Tr(C_\phi M_\lambda^{\rm var})}{V_C},
\end{equation}
where
\begin{equation}
\label{eq:fixed-variance-filter-again}
M_\lambda^{\rm var}
=
\left(\I+e^{(\I-\lambda C_\phi)/T'}\right)^{-1}.
\end{equation}
Thus the normalized retained variance constraint
\begin{equation}
\label{eq:normalized-var-target}
\Tr(\rho_C M_\lambda^{\rm var})=\vartheta_{\rm var},
\qquad
\vartheta_{\rm var}\in(0,1),
\end{equation}
is exactly the scalar probability equation
\begin{equation}
p_{\rm cond}(\lambda)=\vartheta_{\rm var}.
\end{equation}

Moreover, using the unpostselected centered LCU preparation, the joint probability of successfully preparing the centered covariance branch and accepting the thermal measurement is
\begin{equation}
\label{eq:fixed-variance-joint-prob}
p_{\rm joint}(\lambda)
\coloneqq
\Pr(\mathrm{LCU\ success},\, q>\beta\mid \lambda)
=
\frac14\,\Tr(C_\phi M_\lambda^{\rm var}).
\end{equation}
Therefore the unnormalized retained variance constraint
\begin{equation}
\label{eq:absolute-var-target}
\Tr(C_\phi M_\lambda^{\rm var})=\gamma_{\rm var},
\qquad
\gamma_{\rm var}\in(0,V_C),
\end{equation}
can equivalently be tested by the raw joint-probability equation
\begin{equation}
p_{\rm joint}(\lambda)=\frac{\gamma_{\rm var}}{4}.
\end{equation}
\end{proposition}

\begin{proof}
By Proposition~\ref{prop:modified-thermal measurement}, applying the modified thermal measurement with data operator \(\lambda C_\phi\) gives the measurement effect
\begin{equation}
\left(
\I+
e^{(\beta T_2\I-\lambda C_\phi)/(T_1T_2)}
\right)^{-1}.
\end{equation}
Using \(\beta=1/T_2\) and \(T_1T_2=T'\), this becomes
\begin{equation}
\left(
\I+
e^{(\I-\lambda C_\phi)/T'}
\right)^{-1}
=
M_\lambda^{\rm var}.
\end{equation}
Taking the expectation of this effect on the normalized covariance probe
\(\rho_C=C_\phi/V_C\) gives Eq.~\eqref{eq:fixed-variance-conditional-prob}.

For the joint-probability statement, recall that the centered LCU preparation success probability is
\begin{equation}
P_{\rm succ}
=
\frac14\,\Tr(C_\phi)
=
\frac{V_C}{4},
\end{equation}
and the conditioned state is \(\rho_C=C_\phi/V_C\). Hence
\begin{equation}
p_{\rm joint}(\lambda)
=
P_{\rm succ}\,
\Pr(q>\beta\mid \rho_C,\lambda)
=
\frac{V_C}{4}\,
\Tr(\rho_C M_\lambda^{\rm var})
=
\frac14\Tr(C_\phi M_\lambda^{\rm var}),
\end{equation}
which proves Eq.~\eqref{eq:fixed-variance-joint-prob}.
\end{proof}

Since
\begin{equation}
\frac{d}{d\lambda}
\Tr(C_\phi M_\lambda^{\rm var})
=
\frac{1}{T'}
\Tr\!\left[
C_\phi^2 M_\lambda^{\rm var}(\I-M_\lambda^{\rm var})
\right]
\ge 0,
\end{equation}
and the derivative is strictly positive on the support of \(C_\phi\), both
\(p_{\rm cond}(\lambda)\) and \(p_{\rm joint}(\lambda)\) are monotone increasing functions of \(\lambda\). Moreover,
\begin{equation}
\lim_{\lambda\to-\infty}p_{\rm cond}(\lambda)=0,
\qquad
\lim_{\lambda\to+\infty}p_{\rm cond}(\lambda)=1.
\end{equation}
Therefore, for every normalized target
\(\vartheta_{\rm var}\in(0,1)\), there is a unique multiplier \(\lambda^\star\) satisfying
\begin{equation}
p_{\rm cond}(\lambda^\star)=\vartheta_{\rm var}.
\end{equation}
Equivalently, for every feasible unnormalized target
\(\gamma_{\rm var}\in(0,V_C)\), the same multiplier can be identified from
\begin{equation}
p_{\rm joint}(\lambda^\star)=\frac{\gamma_{\rm var}}{4}.
\end{equation}
Thus \(\lambda^\star\) can be found by a one-dimensional classical bisection search.

\begin{proposition}[Sample complexity for fixed-variance calibration]
\label{prop:fixed-variance-complexity}
Fix a bounded bisection interval containing \(\lambda^\star\), and suppose
\begin{equation}
|\lambda|\le \Lambda
\end{equation}
throughout the search.

For the normalized retained variance target
\(\vartheta_{\rm var}\), to obtain
\begin{equation}
\left|
\Tr(\rho_C M_{\widehat\lambda}^{\rm var})-\vartheta_{\rm var}
\right|
\le
\eta_{\rm var}
\end{equation}
with failure probability at most \(\delta\), it suffices to use
\begin{equation}
\label{eq:successful-sample-complexity-var}
S_{\rm succ}
=
O\!\left(
\frac{\log(1/\eta_{\rm var})}{\eta_{\rm var}^2}
\log\frac{\log(1/\eta_{\rm var})}{\delta}
\right)
\end{equation}
successful covariance-probe samples. This bound is independent of both the feature dimension \(d\) and the unknown normalization \(V_C=1-\alpha\). If one instead counts raw centered-LCU preparation attempts, then the expected number of raw attempts is larger by the preparation factor \(1/P_{\rm succ}=\Theta(1/V_C)\), giving
\begin{equation}
\label{eq:raw-sample-complexity-var}
S_{\rm raw}
=
O\!\left(
\frac{1}{1-\alpha}
\frac{\log(1/\eta_{\rm var})}{\eta_{\rm var}^2}
\log\frac{\log(1/\eta_{\rm var})}{\delta}
\right).
\end{equation}

For the unnormalized retained variance target \(\gamma_{\rm var}\), one may instead estimate the raw joint probability \(p_{\rm joint}(\lambda)\). To obtain
\begin{equation}
\left|
\Tr(C_\phi M_{\widehat\lambda}^{\rm var})-\gamma_{\rm var}
\right|
\le
\eta_{\rm var},
\end{equation}
it suffices to estimate \(p_{\rm joint}(\lambda)\) to additive accuracy
\(\eta_{\rm var}/4\). Thus the raw-shot sample complexity is
\begin{equation}
\label{eq:absolute-variance-raw-complexity}
S_{\rm abs}
=
O\!\left(
\frac{\log(1/\eta_{\rm var})}{\eta_{\rm var}^2}
\log\frac{\log(1/\eta_{\rm var})}{\delta}
\right),
\end{equation}
with no additional division by \(1-\alpha\). In this case the unknown normalization is absorbed into the joint Born probability itself.
\end{proposition}

\begin{proof}
For the normalized target, each successful covariance-probe sample produces a Bernoulli random variable with mean
\begin{equation}
p_{\rm cond}(\lambda)
=
\Tr(\rho_C M_\lambda^{\rm var}).
\end{equation}
Hoeffding's inequality implies that
\begin{equation}
O\!\left(
\frac{1}{\eta_{\rm var}^2}
\log\frac{N_{\rm iter}}{\delta}
\right)
\end{equation}
successful samples per bisection step estimate this probability to additive accuracy \(O(\eta_{\rm var})\), uniformly over \(N_{\rm iter}\) bisection steps by a union bound. Since the search is one-dimensional and monotone, a bounded-interval bisection requires
\begin{equation}
N_{\rm iter}=O(\log(1/\eta_{\rm var}))
\end{equation}
steps, up to constants depending on the chosen interval and the local slope of \(p_{\rm cond}\). Multiplying the per-step cost by \(N_{\rm iter}\) gives Eq.~\eqref{eq:successful-sample-complexity-var}. The raw LCU bound Eq.~\eqref{eq:raw-sample-complexity-var} follows from the centered-branch success probability
\begin{equation}
P_{\rm succ}=\frac14 V_C=\frac14(1-\alpha).
\end{equation}

For the unnormalized target, Eq.~\eqref{eq:fixed-variance-joint-prob} gives
\begin{equation}
p_{\rm joint}(\lambda)
=
\frac14\Tr(C_\phi M_\lambda^{\rm var}).
\end{equation}
Therefore an additive error \(\eta_{\rm var}/4\) in estimating \(p_{\rm joint}\) gives an additive error \(\eta_{\rm var}\) in the unnormalized retained variance. Applying the same Hoeffding and union-bound argument to the raw Bernoulli event
\begin{equation}
\{\mathrm{LCU\ success},\,q>\beta\}
\end{equation}
gives Eq.~\eqref{eq:absolute-variance-raw-complexity}.
\end{proof}

The normalized target
\begin{equation}
\vartheta_{\rm var}
=
\Tr(\rho_C M_\lambda^{\rm var})
=
\frac{\Tr(C_\phi M_\lambda^{\rm var})}{\Tr(C_\phi)}
\end{equation}
is usually the more practical explained-variance criterion, because it is a fraction in \((0,1)\) and does not require prior knowledge of
\(\alpha=\braket m m\) or \(V_C=\Tr(C_\phi)\). The algorithm estimates this quantity directly as the conditional probability
\begin{equation}
p_{\rm cond}(\lambda)=\vartheta_{\rm var}.
\end{equation}
The unnormalized target
\begin{equation}
\gamma_{\rm var}=\Tr(C_\phi M_\lambda^{\rm var})
\end{equation}
is also accessible without estimating \(\alpha\): it is encoded in the raw joint probability
\begin{equation}
p_{\rm joint}(\lambda)=\gamma_{\rm var}/4.
\end{equation}
Thus the two targets differ only by whether one conditions on successful covariance-probe preparation. Conditioning gives the normalized explained-variance fraction, while the unconditioned joint event gives the absolute retained covariance weight.

Finally, for circuit implementation, the complexity also needs to include the per-sample quantum cost $\widetilde{O}(1/(T_2^2\varepsilon_{\mathrm{sim}}))$, requiring selector-swap substeps for the Hamiltonian simulation according to Corollary~\ref{cor:gate-complexity}.

\end{document}